\pdfoutput=1

\documentclass[11pt]{scrartcl}
\usepackage[utf8]{inputenc}
\usepackage[english]{babel}
\usepackage[a4paper, margin=2.5cm]{geometry}
\usepackage{authblk}

\title{\vspace{-2em}Growth Functions, Rates and Classes of String-Based Multiway Systems}
\author[1]{Yorick Zeschke}
\affil[1]{Junior Research Affiliate, Wolfram Physics Project}
\date{\today}

\usepackage[style=numeric]{biblatex}
\usepackage{graphicx}
\usepackage{csquotes}
\usepackage{amsthm}
\usepackage{amsmath}
\usepackage{amssymb}
\usepackage{svg}
\usepackage[font=small]{caption}
\usepackage[position=b]{subcaption}
\usepackage{wrapfig}
\usepackage{listings}
\usepackage{algorithm2e}
\usepackage{hyperref}

\def\R{\mathbb{R}}

\def\N{\mathbb{N}}

\def\d{\, \mathrm{d}}

\DeclareMathSymbol{\mlq}{\mathord}{operators}{``}
\DeclareMathSymbol{\mrq}{\mathord}{operators}{`'}
\renewcommand{\phi}{\varphi}
\renewcommand{\epsilon}{\varepsilon}

\newcommand{\ceil}[1]{\left \lceil #1 \right \rceil}
\newcommand{\floor}[1]{\left \lfloor #1 \right \rfloor}
\newcommand{\diffrac}[2]{\frac{\mathrm{d}#1}{\mathrm{d}#2}}

\newtheorem{theorem}{Theorem}[section]
\newtheorem{corollary}{Corollary}[theorem]
\theoremstyle{definition}
\newtheorem{lemma}{Lemma}[section]
\newtheorem{definition}{Definition}[section]

\begin{document}

\maketitle

\begin{abstract}
    In context of the Wolfram Physics Project \cite{wppWebsite}, a certain class of abstract rewrite systems known as \enquote{multiway systems} have played an important role in discrete models of spacetime and quantum mechanics. However, as abstract mathematical entities, these rewrite systems are interesting in their own right. This paper undertakes the effort to establish computational properties of multiway systems. Specifically, we investigate growth rates and growth classes of string-based multiway systems. After introducing the concepts of \enquote{growth functions}, \enquote{growth rates} and \enquote{growth classes} to quantify a system’s state-space growth over \enquote{time} (successive steps of evolution) on different levels of precision, we use them to show that multiway systems can, in a specific sense, grow slower than all computable functions while never exceeding the growth rate of exponential functions. In addition, we start developing a classification scheme for multiway systems based on their growth class. Furthermore, we find that multiway growth functions are not trivially regular but instead \enquote{computationally diverse}, meaning that they are capable of computing or approximating various commonly encountered mathematical functions. We discuss several implications of these properties as well as their physical relevance. Apart from that, we present and exemplify methods for explicitly constructing multiway systems to yield desired growth functions.
\end{abstract}

\tableofcontents

\section{Introduction and Overview}

In 2019, Stephen Wolfram et. al. launched the Wolfram Physics Project \cite{wppWebsite} as a new attempt to find a fundamental theory of physics (see \cite{wol4} and \cite{wol3} for a general overview and \cite{wol1} for a technical introduction, as well as the glossary of \cite{zxPaper} as a reference for terminology used in this paper). The Wolfram Models, discrete spacetime formalisms generalising models first introduced by Wolfram in \cite{wol2}, have been found to be significantly meaningful in a theoretical physics context, showing various connections to known theories of relativity, gravity and quantum mechanics \cite{gor2}\cite{gor1}\cite{zxPaper}.\par 
In the project, a type of abstract rewriting systems (see definition 1 in \cite{gor1}, as well as \cite{rewrite1} and \cite{rewrite2} for more complete references) equipped with causal relations between their elements has been called \enquote{multiway systems} (definition 10 in \cite{gor1}) and shown to be connected to many physical properties of our universe. Additionally, various links between multiway systems and group theory \cite{groupMultiwayWinterSchool}, homotopic type theory \cite{xerxesSummerSchool}, category theory \cite{zxPaper}, numerics of partial differential equations \cite{pdeSummerSchool} and the theory of theorem proving (\cite{wol2} p.\,775\,ff.), as well as the study of complex systems, computational complexity and emergence (\cite{wol1} p.\,204, 939) have been established, showing that these systems are relevant and interesting from a mathematics or computer science perspective as well. While in the physical framework of the actual Wolfram Model (see section 2 in \cite{gor2} for a formal definition), hypergraph-based\footnote{This means that the underlying \enquote{objects} or \enquote{elements} of the abstract rewriting systems are hypergraphs.} multiway systems have been used, we consider string-based multiway systems instead since they are more fundamental because of their simpler structure. Most likely, our results generalise easily to hypergraph-based multiway systems. In any case, they yield significant new insights into the general principles underlying these systems.\par 
Although our investigations are rather theoretical and aim at laying a mathematical foundation for understanding the structure of multiway systems in themselves, we comment on several potential applications in the Wolfram Physics Project in section \ref{section:conclusion} and demonstrate our general result by computational simulations of specific examples. Our visualisations have been made using Mathematica and all code for simulating multiway systems is available in the Wolfram Functions Repository \cite{funcRepo}. Readers interested in running simulations and visualisations themselves may find the computational quick-start guide\footnote{A computational essay can be found at \url{https://www.wolframcloud.com/obj/wolframphysics/Tools/hands-on-introduction-to-the-wolfram-physics-project.nb}.} and the documentation of the \texttt{MultiwaySystem} resource function \cite{multiwayFunc} to be useful references.\par 
The subsequent subsections start by formally defining what we mean by multiway system growth functions, rates and classes. Next, we investigate the boundaries of possible growth rates and find that multiway systems are, simply put, bounded in the speed but unbounded in the slowness of their growth rate (theorem \ref{theorem:slowness}). After that, we show that the growth classes of multiway systems we defined cover the entire set of multiway systems and apart from one trivially empty class, all of them contain infinitely many multiway systems (theorem \ref{theorem:classesPartition}). To do this, we define arithmetic-like operations equipping the set of multiway systems with a semiring structure. Combing the two theorems, we conclude various interesting properties of the \enquote{computational diversity} and \enquote{-complexity} of multiway systems, showing that their growth functions constitute an interesting domain of further research. 

\subsection{Multiway Growth Functions}

Consider a string-based multiway system $M$ (definition 10 in \cite{gor1}), represented as a triplet $(R,s_{\text{init}},\Sigma)$ where $\Sigma$ is a finite alphabet, $R=\{r_1\rightarrow t_1, \dots, r_n\rightarrow t_n\}$ is a set of string replacement rules over $\Sigma$ and $s_{\text{init}} \in \Sigma^*$ is the initial string\footnote{In the following, $\Sigma^*$ denotes the set of all words over the alphabet $\Sigma$.}. We define the \enquote{state-set of generation $n$} as the set of all new (previously nonexistent) states added to the multiway system in its $n$-th generation. These states are precisely the nodes of the states graph (c.\,f. section 5.3 in \cite{wol1}) to which the shortest path from the initial state has length $n$. In \cite{wol1} they are called \enquote{merged states}. Now, the \enquote{growth function\footnote{We use the terms \enquote{sequence} and \enquote{function} interchangeably for $f:\N_+\rightarrow \N_+$.}} $g_M(n)$ is simply the cardinality of the state-set of generation $n$.\par 

\begin{figure}[ht]
\centering
\begin{subfigure}{.5\textwidth}
  \centering
  \includegraphics[width=.9\linewidth]{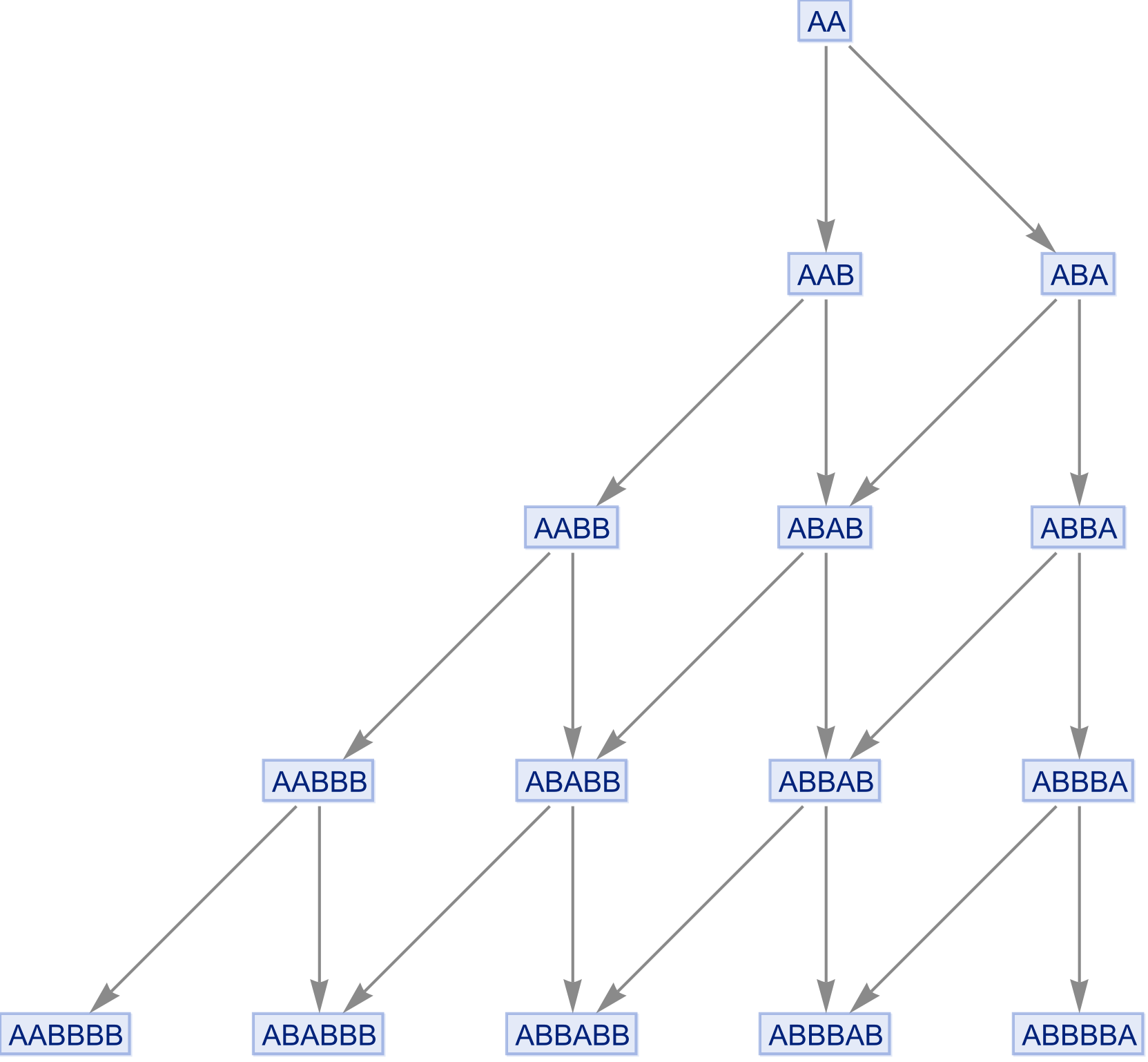}
\end{subfigure}%
\begin{subfigure}{.5\textwidth}
  \centering
  \includegraphics[width=.9\linewidth]{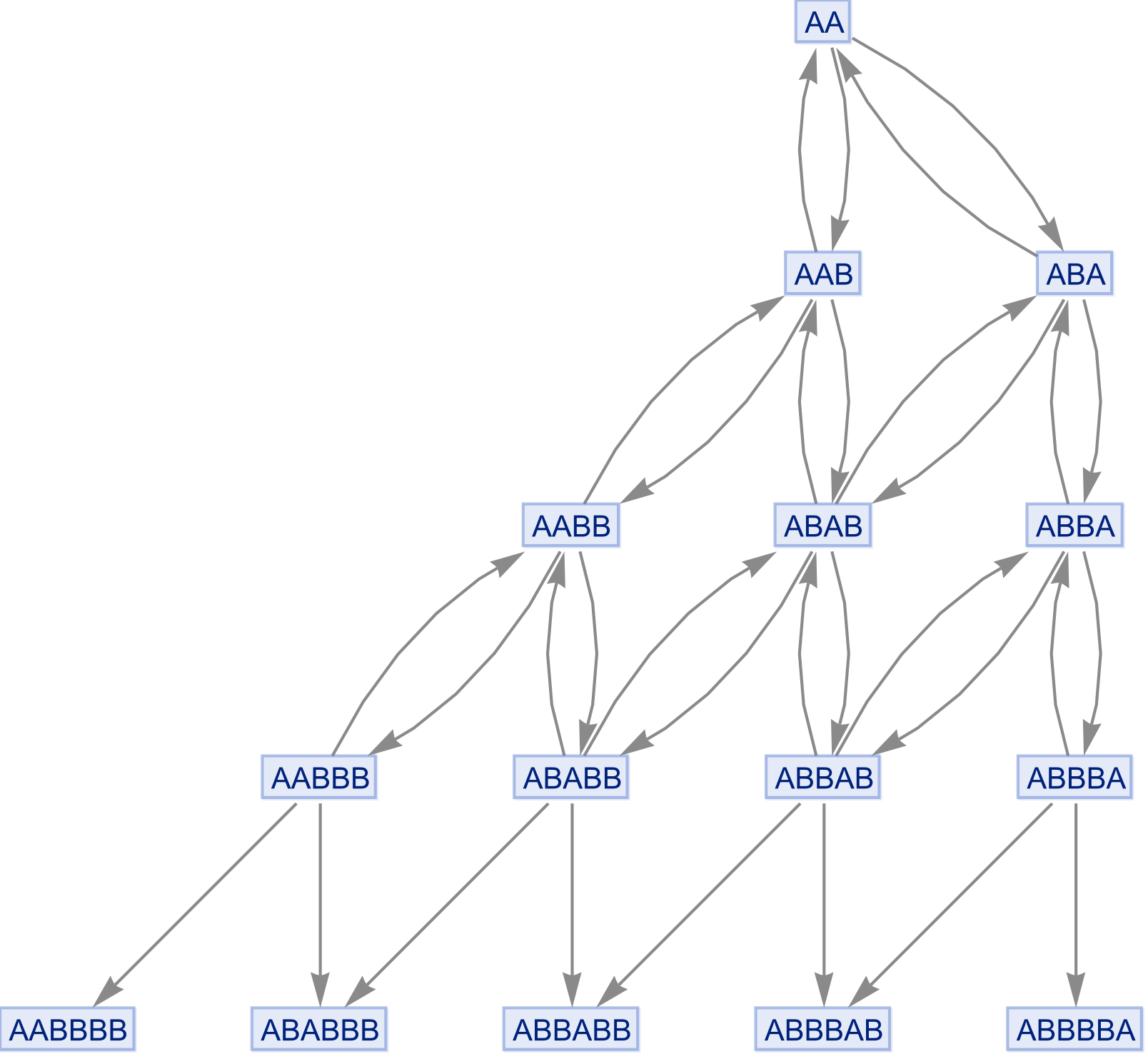}
\end{subfigure}
\caption{Both $M_1=(\{``A"\rightarrow ``AB"\}, ``AA", \{A,B\})$ and $M_2=(\{``A"\rightarrow ``AB",``AB"\rightarrow ``A"\}, ``AA",\{A,B\})$ have the same growth function $g(n)=n$ because cycles in the states graph do not lead to new states.}
\label{fig:stateSet}
\end{figure}

In general, it is very hard or even undecidable (see section \ref{section:consequencesComputability}) to prove that some multiway system has a certain growth function. It is also not obvious that the growth functions of multiway systems should be elementary functions or \enquote{simple} by any other definition. Examples of systems where the growth function is hard to describe were already given in \enquote{A New Kind of Science} (\cite{wol2}\,pp.\,204\,ff.). Therefore, we will approximate the growth functions of multiway systems by continuous, strictly monotonically increasing, unbounded (and hence bijective on $\R_{\geq 0}$) functions which can be analysed more easily. This way, many similar growth functions will be considered members of the same equivalence class. We will say that the corresponding multiway systems have the same \enquote{growth rate}.

\subsection{Multiway Growth Rates}

To formalize the notion of approximating functions, we use the asymptotic growth classes from complexity theory, defined in the following way: For a function $f:A\rightarrow B$ where $A=B=\R_{\geq 0}$ or $A=B=\R$, $\mathcal{O}(f)$ is defined as the set of functions $g:A\rightarrow B$ for which $\limsup_{x\rightarrow \infty} |\frac{g(x)}{f(x)}|$ exists and is a real number\footnote{An equivalent definition is $g\in \mathcal{O}(f)\iff \exists C\in \R_+: \exists x_0\in A: \forall x>x_0: C|f(x)|\geq |g(x)|$.}. The subset of $\mathcal{O}(f)$ for which the limit superior is zero is denoted $o(f)$. Similarly, $\Omega(f):=\{g:A\rightarrow B \mid f\in \mathcal{O}(g) \}$ and $\omega(f):=\{g:A\rightarrow B \mid f\in o(g) \}$. Finally, $\Theta(f):=\Omega(f)\cap \mathcal{O}(f)$. It is straightforward to show that $f\sim_{\Theta} g \iff f\in \Theta(g)$ is an equivalence relation. Thus, we may speak of functions that are \enquote{asymptotically equal}.\par 
As mentioned above, we want to approximate growth functions by bijective functions for the subsequent mathematical analysis. For some multiway growth function $a$, we will define the sequences $\overline{a}$ and $\underline{a}$ as its tightest upper and lower bounds which are monotonically increasing, even if $a$ itself is not monotonic at all. From these sequences, we will then construct two equivalence classes of continuous functions which are all asymptotically equal to and hence \enquote{close approximations} of $\overline{a}$ or $\underline{a}$ respectively. Two representatives of these classes will be called \enquote{tight bounds} and, since we are generally concerned with unbounded growth functions\footnote{Bounded growth functions will be discussed shortly.}, both bijective on $\R_{\geq 0}$.\par 
Notice that we have only defined asymptotic growth classes for functions on $\R$ or $\R_{\geq 0}$. However, since the multiway growth function is always a function on $\N_+$, we consider its linear interpolation, a continuous function from $\R_{\geq 0}$ to $\R_{\geq 0}$ which is equal to the sequence for natural arguments and always bounded by consecutive values of the sequence (see definition \ref{def:linearInterpolation}), instead.

\begin{definition}\label{def:finiteMultiway}
Let $M$ be a multiway system and $g_M$ its growth function. We call $M$ \enquote{finite} if $\exists n\in \N_+: g_M(n)=0$ (as this implies that at a certain point, no further states will be added). We call $M$ \enquote{bounded} if $\exists b\in \N_+: \forall n\in \N_+: g_M(n) \leq b$ and $M$ is not finite. Systems which are neither finite nor bounded are called \enquote{unbounded}.
\end{definition}

\begin{figure}[ht]
\centering
\begin{subfigure}{.5\textwidth}
  \centering
  \includegraphics[width=.9\linewidth]{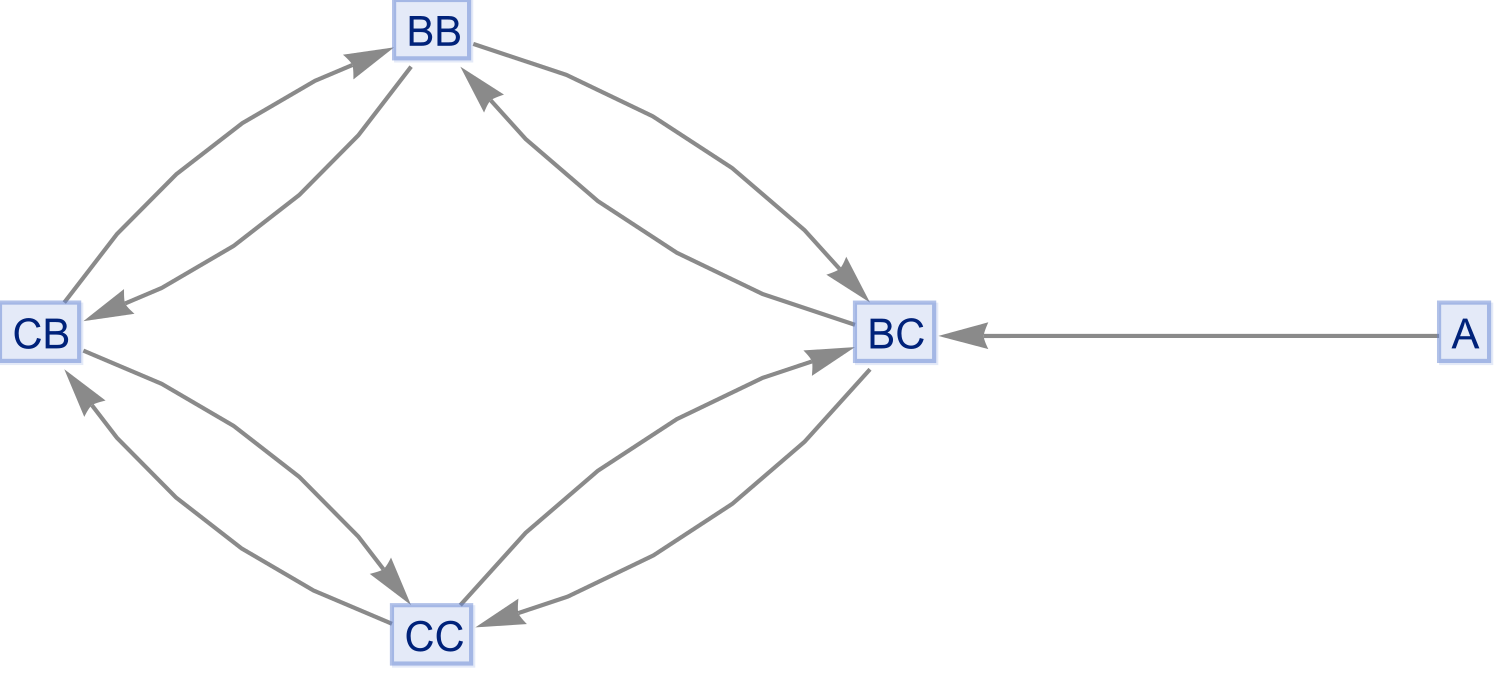}
\end{subfigure}%
\begin{subfigure}{.5\textwidth}
  \centering
  \includegraphics[width=.9\linewidth]{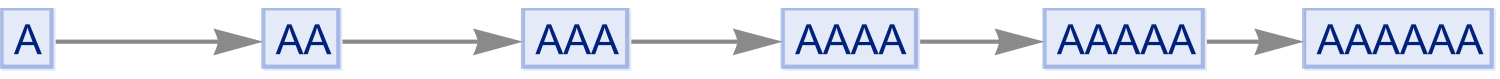}
\end{subfigure}
\caption{States graphs of the finite system $M_1=(\{ ``A"\rightarrow ``BC", ``B"\rightarrow ``C", ``C"\rightarrow ``B" \},``A",\{A,B\})$ and first six steps of the bounded system $M_2=(\{``A"\rightarrow ``AA"\}, ``A", \{A\})$. Notice that $\forall n\in \N_+: g_{M_2}(n)=1$, despite the fact that the rule can be applied in many different positions, because we are only considering merged states.}
\label{fig:finiteAndBounded}
\end{figure}

\begin{definition}\label{def:tightBounds}
Let $a:\N_+\rightarrow \N_+$ be the growth function of an unbounded multiway system and let $\overline{a}_n:=\max( \{a_k \mid k\leq n\})$ and $\underline{a}_n:=\max( \{ a_k \mid  k \leq n\,\land\,(\forall l\geq k: a_l\geq a_k ) \} \cup \{ 1 \})$.
We call two continuous functions $f,g:\R_{\geq 0}\rightarrow \R_{\geq 0}$ \enquote{tight bounds of $a$} if $f\in \Theta(L_{\N_+}(\overline{a}))\;\land\;g\in \Theta(L_{\N_+}(\underline{a}))$ where $L_{\N_+}$ denotes the linear interpolation over $\N_+$ (according to definition \ref{def:linearInterpolation}).
\end{definition}

\begin{figure}[ht]
\centering
\begin{subfigure}{.5\textwidth}
  \centering
  \includegraphics[width=.9\linewidth]{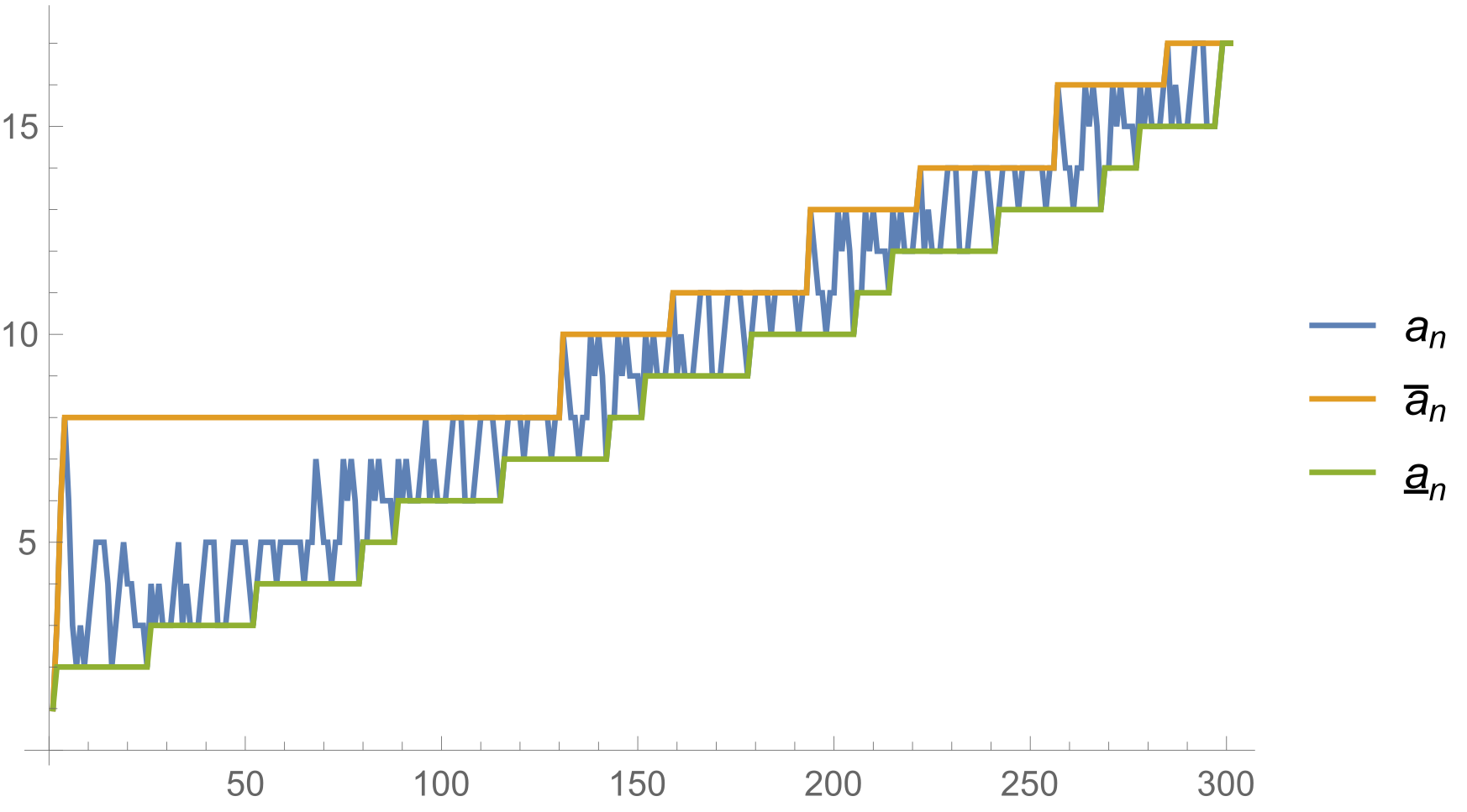}
\end{subfigure}%
\begin{subfigure}{.5\textwidth}
  \centering
  \includegraphics[width=.9\linewidth]{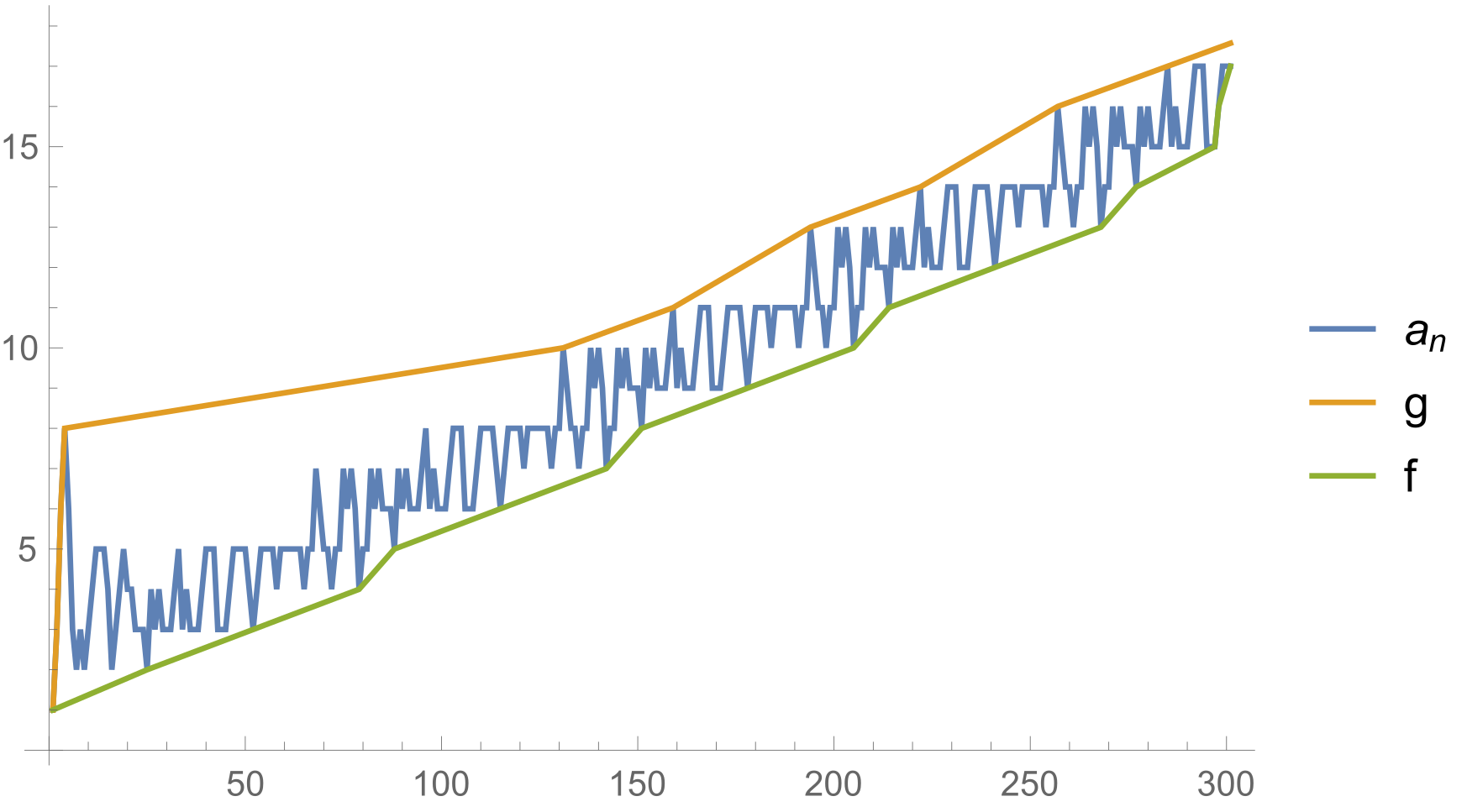}
\end{subfigure}
\caption{The growth function $a_n$ of $M=(\{``AB" \rightarrow ``",``ABA"\rightarrow ``ABBAB",``ABABBB"\rightarrow ``AAAAABA"\}, ``ABABAB", \{A,B\})$ together with $\overline{a}_n$, $\underline{a}_n$ and a pair $(f,g)$ of tight bounds. Note: Only $f$ and $g$ are continuous, the other lines are drawn for visual appearance.} 
\label{fig:tightBounds}
\end{figure}

One might ask why we introduce an upper and a lower bound instead of approximating the growth function with a single function. Doing so would however be a poor approximation as there are multiway systems for which even the tightest upper and lower bounds are never in the same asymptotic growth class (compare figure \ref{fig:notSameThetaBounds}). We will call these multiway systems \enquote{strongly oscillating} and all others (i.\,e. systems where all tight bounds are asymptotically equal) \enquote{regular}. Notice that for every regular multiway system, a pair of bijective tight bounds exists because its tight bounds will be in the asymptotic equivalence class of two unbounded strictly monotonically increasing functions and tight bounds are continuous on $\R_+$. Strongly oscillating systems on the other hand are much more difficult to analyse since one cannot easily come up with criteria for measuring the rate of oscillation and it is not clear at all whether there has to be any periodicity or regularity in the way in which they oscillate. Thus, for our basic investigations about the fundamental structure of multiway systems, we shall focus on regular systems.

\begin{figure}[ht]
\centering
\begin{subfigure}{.45\textwidth}
  \centering
  \includegraphics[width=.9\linewidth]{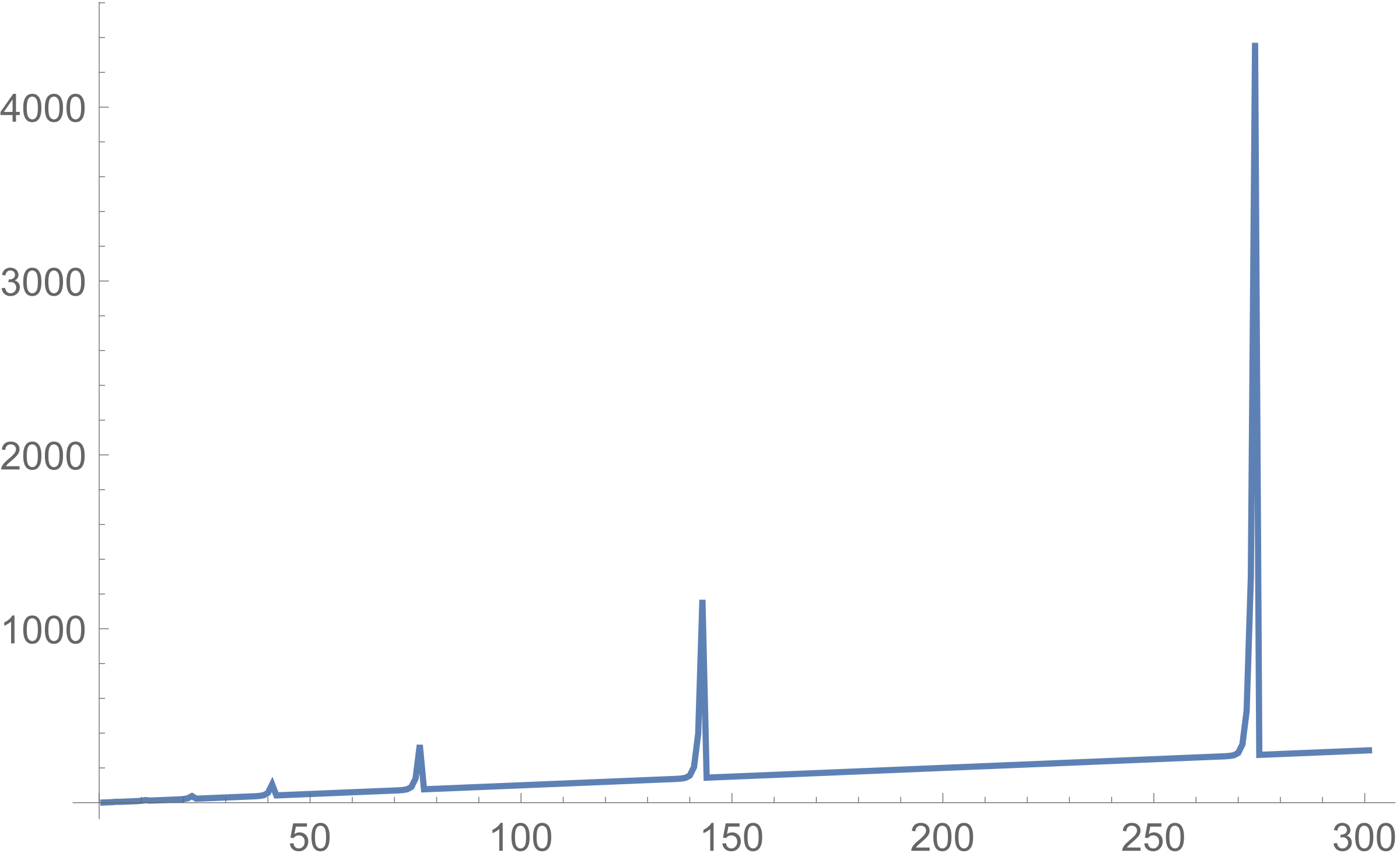}
\end{subfigure}%
\begin{subfigure}{.55\textwidth}
  \centering
  \includegraphics[width=.9\linewidth]{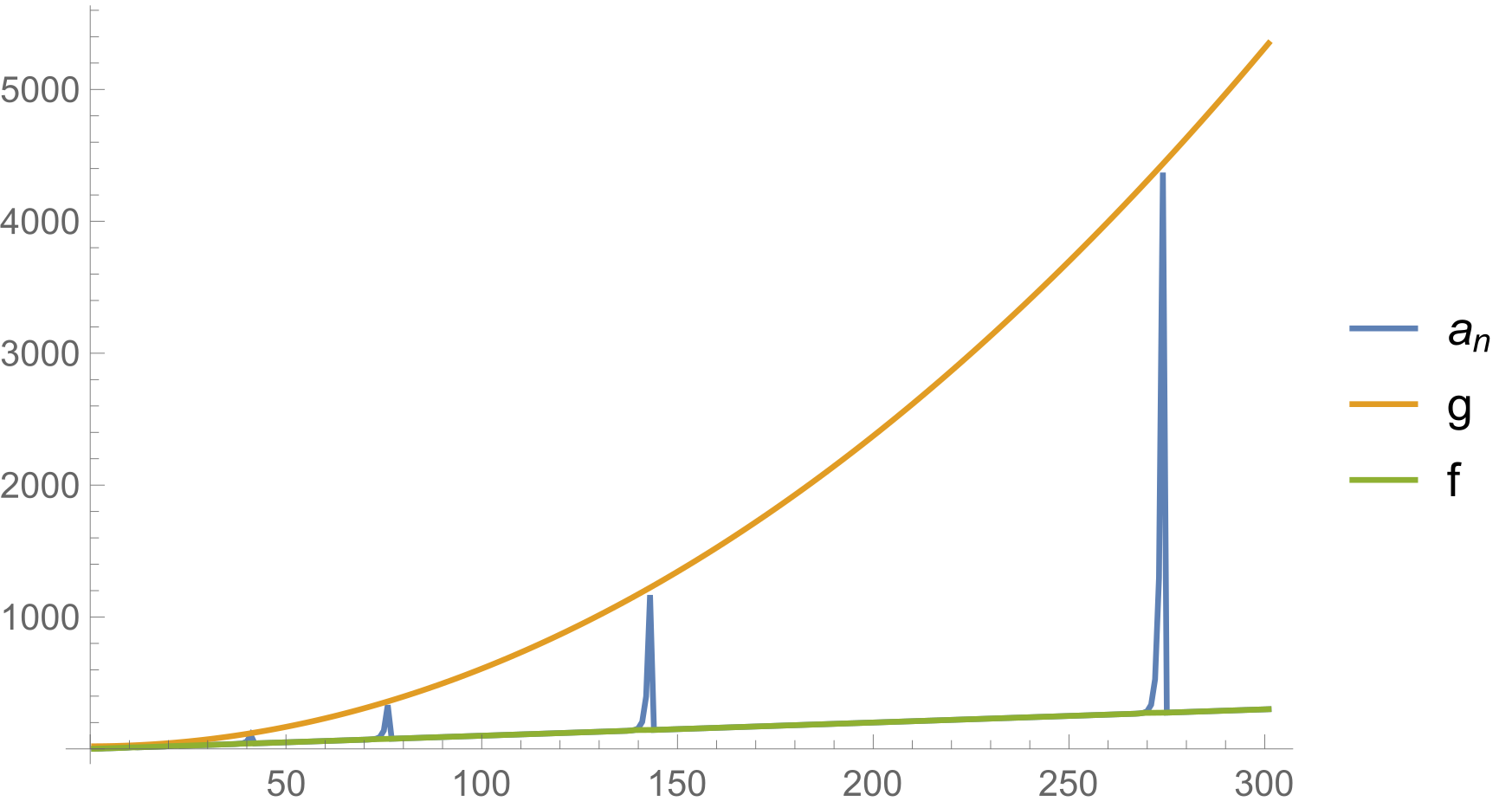}
\end{subfigure}
\caption{For this special system, one can prove (see section \ref{section:consequencesComputability} for an explanation) the tight bounds $f,g$ to be in different asymptotic growth classes: $f\in \Theta(x)$ and $g\in \Theta(x^2)$. Note again that only $f$ and $g$ are continuous while $a_n$ is drawn as a line for visual appearance.}
\label{fig:notSameThetaBounds} 
\end{figure}

Definition \ref{def:tightBounds} suggests a natural way to define classes of multiway systems with \enquote{similar} growth functions by considering growth functions with approximately equal tight bounds as equivalent. Let $f,g:\N_+\rightarrow \N_+$ be functions and $(a_1,b_1),(a_2,b_2)$ be tight bounds of $f$ and $g$ respectively. We define $\sim_R$ by $f \sim_R g \iff (a_1\sim_{\Theta} a_2 \land b_1 \sim_{\Theta} b_2)$. Since tight bounds always exist, $\sim_R$ is an equivalence relation because $\sim_{\Theta}$ is one. For some multiway system $M$ with growth function $g_M$, we call the equivalence class of $\sim_R$ that $g_M$ falls into the \enquote{growth rate} of $g_M$ (or sometimes the growth rate of just $M$).\par 
It is obvious that every multiway system has exactly one growth function and exactly one growth rate. The converse, i.\,e. that every function $f:\N_+\rightarrow \N_+$ is the growth function of some multiway system or, respectively, that every pair of bijective functions on $\R_{\geq 0}$ is a pair of tight bounds of some multiway growth function, is clearly not true as emphasized in lemma \ref{lemma:expBound}. However, if we define much more general classes of growth functions which we will call \enquote{multiway growth classes}, we will see (in theorem \ref{theorem:classesPartition}) that they indeed partition the set of all multiway systems into a finite set of infinite subsets.

\subsection{Multiway Growth Classes}\label{section:growthClassesDef}

To further distinguish between types of multiway systems on a more abstract level and demonstrate which kinds of growth functions can be achieved, we want to define very broad classes of multiway systems whose growth functions show similar behavior on a large scale. We have already distinguished between finite, bounded and unbounded systems, as well as dividing the latter into regular and strongly oscillating systems. As outlined above, we will focus on regular systems. To group these into sets of systems of similar behavior, we use commonly known classes of functions such as polynomial or exponential functions\footnote{Actually, we are talking about polynomially or, respectively, exponentially bounded functions which need not be polynomials or simple exponential functions.}, intermediately (faster than polynomial and slower than exponential) growing functions and some others.\par  
More precisely: Let $G_{\text{pol}}$ be defined as the set of all continuous bijections $f:\R_{\geq 0}\rightarrow \R_{\geq 0}$ which satisfy $f\in \Omega(x^n) \cap \mathcal{O}(x^{n+1})$ for some $n\in \N_+$ and define $G_{\text{exp}}$ as $\{ f:\R_{\geq 0}\rightarrow \R_{\geq 0} \mid f\in \Omega(a^x)\cap \mathcal{O}((a+1)^x)\}$ for some  $a\in \N_{> 1}$. Similarly, let $G_{\text{supexp}}$ be the set $\{ f:\R_{\geq 0}\rightarrow \R_{\geq 0} \mid \forall g\in G_{\text{exp}}: f \in \omega(g) \}$ where $f$ must be continuous and bijective. Additionally, denote by $G_{\text{int}}$ the set of all continuous bijections $f:\R_{\geq 0}\rightarrow \R_{\geq 0}$ fulfilling $\forall g\in G_{\text{pol}}, h\in G_{\text{exp}}: f\in \omega(g) \cap o(h)$. Now, it is easy to define $G_{\text{invpol}}:=\{f \mid f^{-1} \in G_{\text{pol}}\}$, $G_{\text{invexp}}:=\{f\mid f^{-1} \in G_{\text{exp}}\}$, $G_{\text{invsupexp}}:=\{f\mid f^{-1}\in G_{\text{supexp}}\}$ and $G_{\text{invint}}:=\{f\mid f^{-1}\in G_{\text{int}}\}$.\par 
These eight sets give a partition of the set of all continuous bijections on $\R_{\geq 0}$ because a function grows either slower than $f:x\mapsto x$ in which case its inverse grows faster than $f$, or it grows faster (or equal to) $f$ in which case it is contained in one of the first four classes. We call a multiway system a member of the growth class $C_{i}$ if its growth function has tight bounds $f,g$ belonging\footnote{This definition applies only to regular multiway systems} to $G_{i}$. However, as definition \ref{def:tightBounds} is not applicable for finite or bounded multiway systems, we handle them separately.\par 
Let $C_{\text{fin}}$ and $C_{\text{bnd}}$ the sets of all finite and bounded multiway systems respectively. For all multiway systems in $C_{\text{fin}} \cup C_{\text{bnd}}$, the growth rate is defined to be $(1,1)$. While finite and bounded systems can have quite an intricate structure\footnote{In fact, what seems to be \enquote{complex behavior} in a \enquote{New Kind of Science}-fashion (compare \cite{wol2}), occurred much more frequently in our empirical investigations of finite systems but this might just indicate a lack of understanding.}, their growth functions are not very interesting for our purposes. They might be useful for applications not directly related to the Wolfram Physics Project, but in this paper, they will not be discussed in great detail. \par
As every continuous bijection on $\R_+$ belongs to exactly one of the $G_i$, every multiway system one can imagine is either strongly oscillating or in one of those classes (including $C_{\text{fin}}$ and $C_{\text{bnd}}$). We will furthermore show in theorem \ref{theorem:classesPartition} that every of these classes (except $C_{\text{supexp}}$ which is empty by lemma \ref{lemma:expBound}) contains infinitely many multiway systems. \par 
Summarizing the previous section, we introduced the three main concepts of multiway growth functions, multiway growth rates and multiway growth classes. We will now present the first important result of this paper, a theorem about the boundaries of possible growth rates, and spend the next section proving and illustrating it.

\section{The Spectrum of Possible Growth Rates}

Having defined multiway growth rates, we might ask ourselves, which growth rates are possible, i.\,e. how the equivalence classes of $\sim_R$ are distributed in the set of all possible pairs of bijective functions on $\R_{\geq 0}$. First of all, it is quite easy to give an upper bound for growth rates that can be achieved. In fact, no multiway system can grow faster than exponentially. 

\begin{lemma}\label{lemma:expBound}
Let $(f,g)$ be the growth rate of some multiway system. There exists some constant $c\in \R$ for which $f, g \in o(e^{cx})$.
\end{lemma}
\begin{proof}
Denote by $s_{max}(n)$ the maximum string length that states of generation $n$ can have. For every multiway system $M=(R,s_{init},\Sigma)$, the set of rules remains constant during the whole evolution, so $s_{max}$ can at most increase constantly, i.\,e. $s_{max}\in \mathcal{O}(n)$. Since the number of words with length $l$ is given by $|\Sigma|^l$, the growth function $g_M(n)$ will never exceed $|\Sigma|^{s_{max}(n)} = e^{\ln(|\Sigma|) s_{max}(n)} \in \Theta(e^{cn})$ and the claim follows.
\end{proof}

So what about a lower bound for multiway growth rates? Formally, a trivial multiway system with no rules and thus only one state has the lowest possible growth function by point-wise value comparison. In general, \enquote{terminating} or \enquote{constant} asymptotic growth functions of finite or bounded multiway systems (which have the growth rate $(1,1)$) are the slowest by means of asymptotic comparison\footnote{\enquote{Asymptotic comparison} refers to the total ordering $\leq_{\mathcal{O}}$ defined by $f \leq_{\mathcal{O}} g \iff f\in \mathcal{O}(g)$}, but examples like this are not very illuminating. Therefore, we might ask what the slowest growth rate faster than constant is, i.\,e. what the smallest (by asymptotic comparison) functions $f,g \in \omega(1)$ are, for which $(f,g)$ is the growth rate of some multiway system. It turns out however, that no such smallest growth rate exists which means that multiway system can, in a certain sense, grow arbitrarily slowly. \par 
To understand why this is the case and make the even stronger statement that multiway systems can grow slower than all computable functions (see corollary \ref{corollary:multiwaySlowerComputable}), we need to introduce a couple of constructions. First of all, we will show how multiway systems can emulate Turing machines, meaning there are systems such that the successive states of their evolution correspond to steps in the machine's evaluation. Phrased differently, it is possible to construct a multiway system which has exactly one new state for $T(n)$ steps where $T(n)$ is the number of operations that a certain Turing machine $\mathcal{T}$ carries out before halting when provided with the input $n$. Such a machine, with some additional constraints, will be called a \enquote{$T$-halter} and $T$ its \enquote{halting function}. \par 
By adding some specific rules to the multiway system that emulates $\mathcal{T}$, it will be possible to evaluate $\mathcal{T}$ indefinitely for increasing inputs $n=1,2,\dots$. Additionally, the multiway system will be constructed in a way such that every time, the underlying Turing machine is \enquote{started again} on the next input, the number of new states per time step is increased by one. This way, we will obtain a growth function informally described by the sequence \enquote{$n$ occurs $T(n)$ times} (see definition \ref{def:nOccursFtimes}), for example the sequence \enquote{$n$ occurs $n$ times}, which would be given by $1,2,2,3,3,3,4,4,4,4,5,\dots$. We will then show that this growth function is approximated by the inverse of the linear interpolation (see definition \ref{def:linearInterpolation}) over the summatory function of $T$ (see lemma \ref{lemma:inverseSummatory}). From this we conclude the following theorem:

\begin{theorem}\label{theorem:slowness}
Let $T\in\Omega(n)$ be the halting function of some Turing machine. There is a multiway system with growth rate $(a,b)$ such that $a,b\in \mathcal{O}(T_*^{-1})$ where $T_*(x)=L_{\N_+}(\sum_{k=1}^n T(k))$ and $L_{\N_+}$ denotes the linear interpolation over $\N_+$ (see definition \ref{def:linearInterpolation}).
\end{theorem}

\begin{figure}[ht]
\centering
\begin{subfigure}{\textwidth}
  \centering
  \includegraphics[width=.75\linewidth]{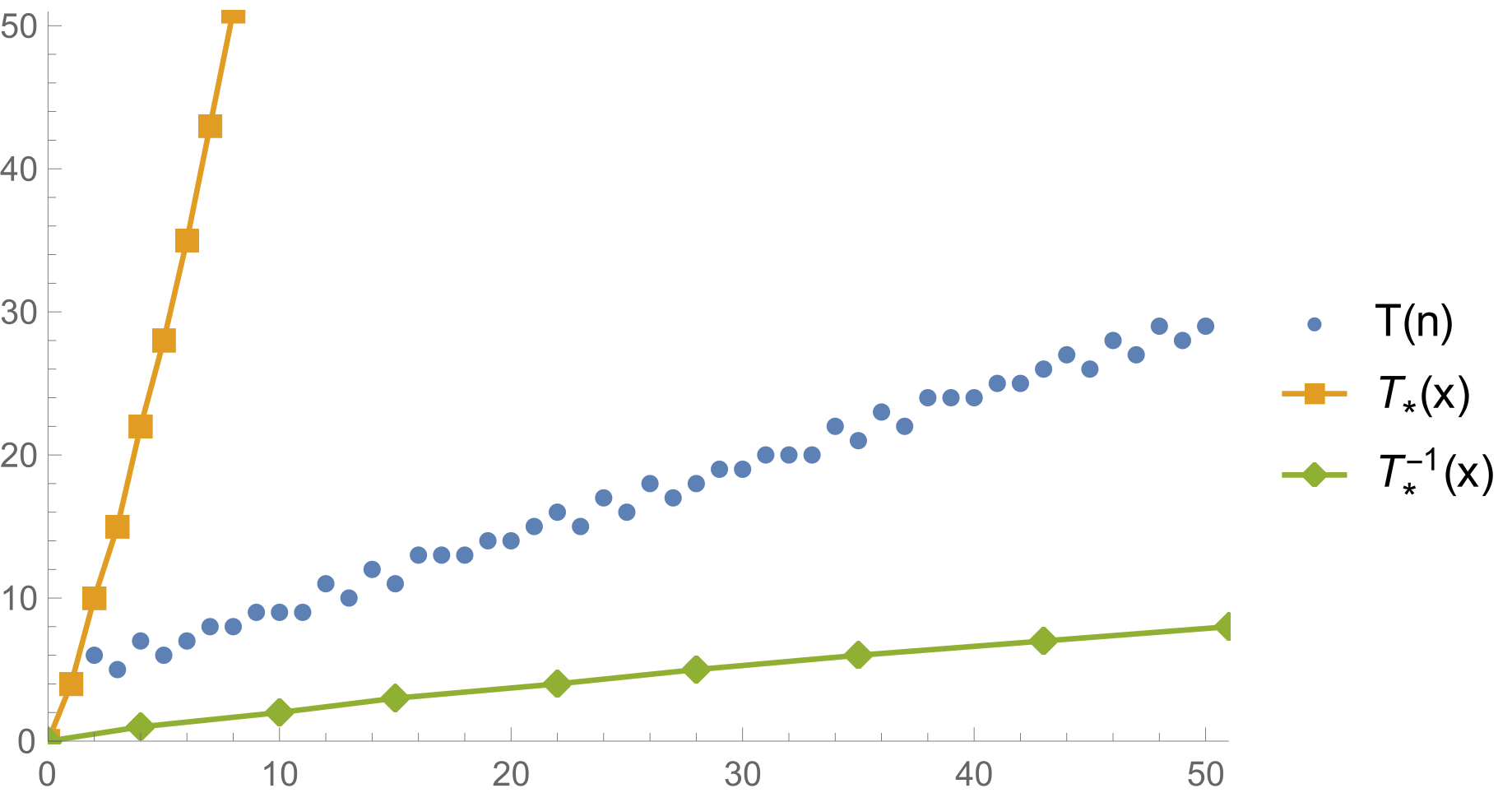}
\end{subfigure}
\caption{Graphical illustration of theorem \ref{theorem:slowness}. The theorem asserts that there is a multiway system for which the growth rate $(f,g)$ is asymptotically less than $T_*^{-1}(x)$.}
\label{fig:theorem1}
\end{figure}

\subsection{Proof of Theorem \ref{theorem:slowness}}\label{section:slownessProof}

Now let us formalize the proof outlined above. For some function $T:\N \rightarrow \N$, we define a \enquote{$T$-halter} to be a Turing machine $\mathcal{T}$ such that $\mathcal{T}$ executes precisely $T(n)$ operations when given the input $n$, taking into account the input and output constraints depicted in figure \ref{fig:tHalterTape}. These constraints will later allow us to \enquote{enchain} the multiway systems corresponding to $T$-halters. Of course, neither is there a $T$-halter for arbitrary $T$ nor must there be a unique $T$-halter for a given $T$. However, and this is the part we care about, there is a $T$-halter such that $T\in \Omega(f)$ for any computable function $f$ because we can just take the Turing machine that computes $f$ and add some logic to write $n+1$ after the computation is finished.\par
Having defined $T$-halters, the next step is to show how multiway systems can emulate these (and all other Turing machines). Since we are talking about deterministic Turing machines, no branching shall occur in the corresponding multiway system, i.\,e. the system should have exactly one state in generation $n$ which corresponds to the state of the Turing machine after $n-1$ operations. Because at every state in the machine's evolution only a finite part of the tape contains non-blank symbols, we include only the symbols already \enquote{touched} by the machine (meaning the head was on that symbol at least once) in the states of the multiway system and abbreviate the infinite strings of zeros on both sides of the tape with an underscore. The position and state of the head are indicated by an $H$ right of the symbol the head is currently on, followed by the current state number. Hence, there are four additional symbols (two underscores, an $H$ and a number) used in the multiway system but not written on the machine's tape (see figure \ref{fig:turingMultiwayTapeStates}). \par

\begin{figure}[ht]
\centering
\begin{subfigure}{.45\textwidth}
  \centering
  \includegraphics[height=15em]{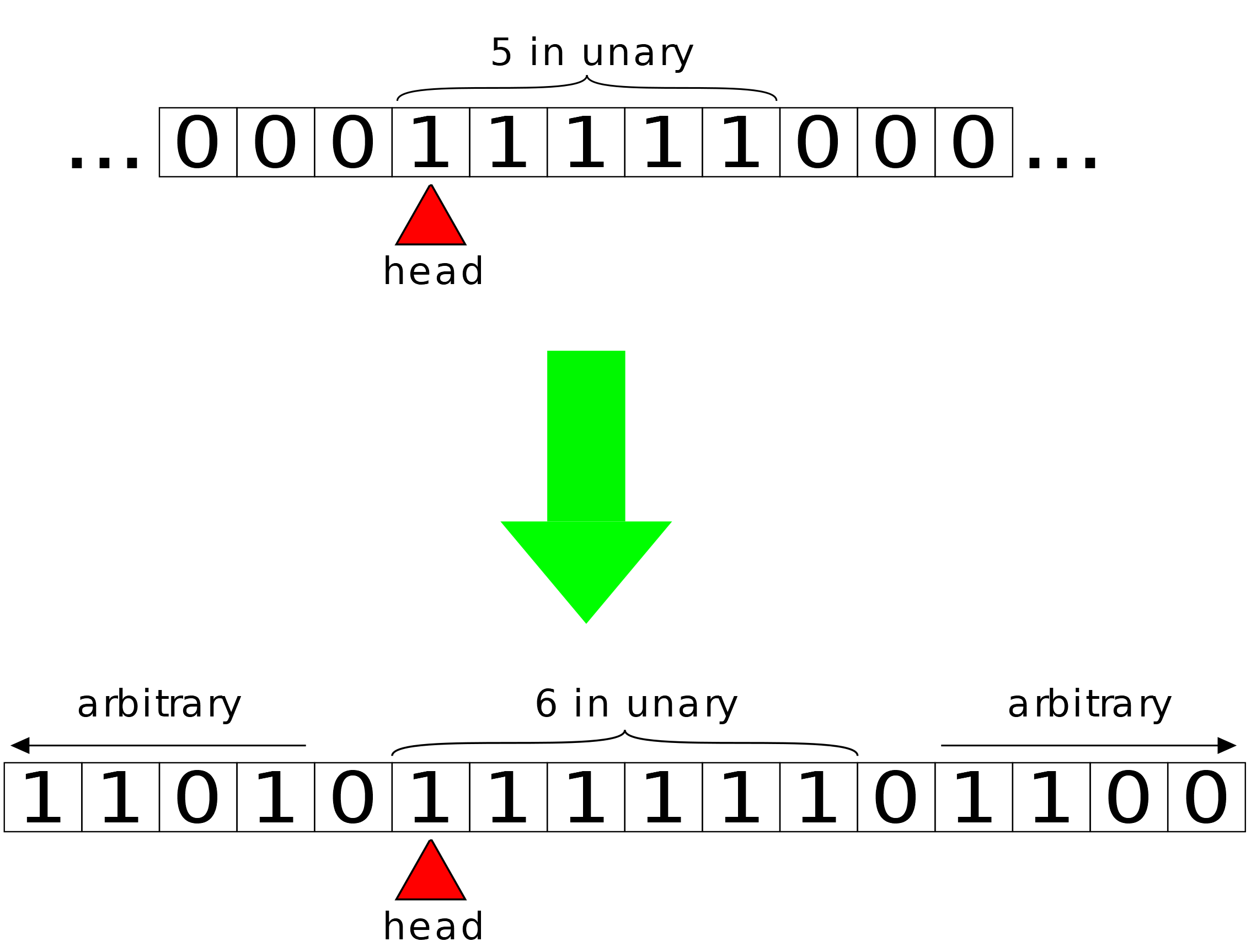}
    \caption{$\mathcal{T}$ is always started on an empty tape containing only $n$ in unary representation. After halting, $\mathcal{T}$ is required to have written $n+1$ in unary on the tape and placed its head onto or left of the first digit. The number $n+1$ must be preceded and followed by at least one empty symbol.}
\label{fig:tHalterTape}
\end{subfigure}
\hfill
\begin{subfigure}{.45\textwidth}
  \centering
  \includegraphics[height=15em]{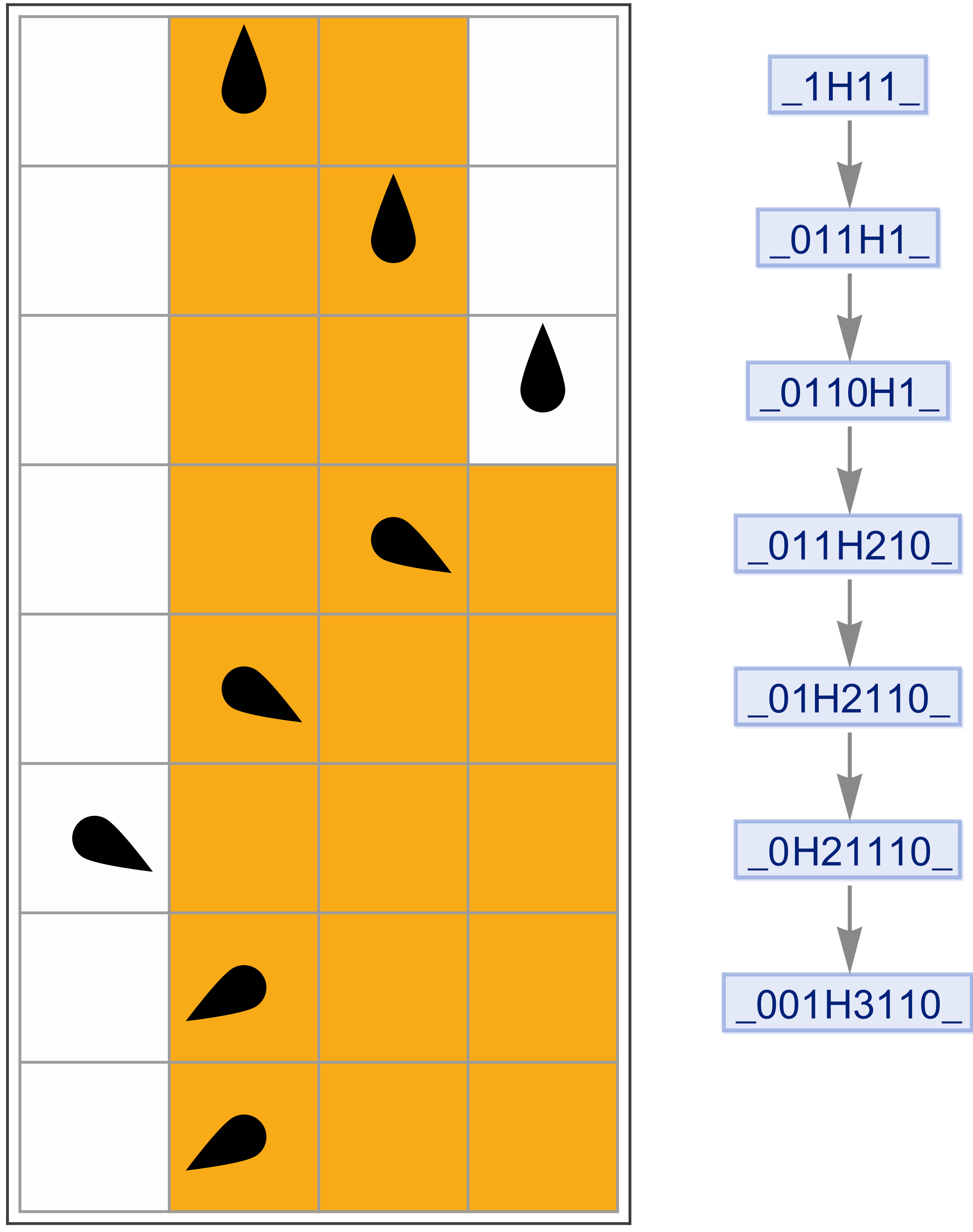}
  \caption{Evolution of $\mathcal{T}_1$ (compare fig.\,\ref{fig:turingExample}) next to the evolution of $M_1$.}
  \label{fig:turingMultiwayTapeStates}
\end{subfigure}%
\caption{}
\end{figure}

Using this representation, plain read/write operations and state changes would be straightforward to implement as replacement rules, as we could just introduce a rule\footnote{Writing symbols next to each other in this contexts simply denotes their concatenation to a string.} \texttt{$``x H n" \rightarrow ``y H m"$} for every combination of currently read symbol $x$ and head state $n$. However, since the head must move left or right after each such operation, those rules are not suitable. What is needed instead to encode the operation \enquote{when $x_1$ is read in state $n$, write $y_1$, change state to $m$ and move the head right}, is a rule of the form \texttt{$``x_1 H n x_2" \rightarrow ``y_1 x_2 H m"$} for every possible value of $x_2$. Similarly, a left move of the head is encoded as \texttt{$``x_2 x_1 H n" \rightarrow ``x_2 H m y_1"$}. If $n$ is not a halting state (in which case we would not need any rules), exactly one of these two rule patterns will be applicable for every possible value of $x_1$ or, respectively, every state transition arrow starting at $H$ in the state transition diagram.\par  
For a Turing machine with $N$ states working on an alphabet of $S$ symbols, that already gives a worst-case (no halting states) of $N\cdot S^2$ rules in the corresponding multiway system. However, $N\cdot S$ (worst-case) more rules have to be added to handle the literal \enquote{edge cases} in which the head is next to one of the underscores bounding the tape. The two rule patterns, of which, as before, exactly one will match for every state transition arrow, are \texttt{$``\_ x H n" \rightarrow ``\_ 0 H m y"$} for a left move and \texttt{$``x H n \_"\rightarrow ``y 0 H m \_"$} for a right move (in both cases $x$ is read and $y$ written). Now, the resulting rule set captures all of the Turing machine's properties and is able to extend the tape to any required length by itself. As an initial state of the multiway system to emulate the machine, any string of characters from the machines alphabet together with an $``H s"$ where $s$ is the starting state and the two bounding underscores can be used. \par 
To illustrate this construction, consider the Turing machine $\mathcal{T}_1$ shown in figure \ref{fig:turingExample}. It is, in some sense, the easiest possible $T$-halter for it does nothing more than increasing the number on the tape by one and placing its head back at the beginning. Figure \ref{fig:turingMultiwayTapeStates} shows the successive states of the machine (and tape) next to a multiway system $M_1$ emulating $\mathcal{T}_1$. The rule set for this specific instance is
\begin{verbatim}
{"00H1" -> "0H21", "10H1" -> "1H21", "1H10" -> "10H1", "1H11" -> "11H1",
 "01H2" -> "0H21", "11H2" -> "1H21", "0H20" -> "00H3", "0H21" -> "01H3", 
 "_0H1" -> "_0H21", "1H1_" -> "10H1_", "_1H2" -> "_0H21", "0H2_" -> "00H3_"}
\end{verbatim} \par

\begin{figure}[ht]
\centering
\begin{subfigure}{.45\textwidth}
  \centering
  \includegraphics[height=16em]{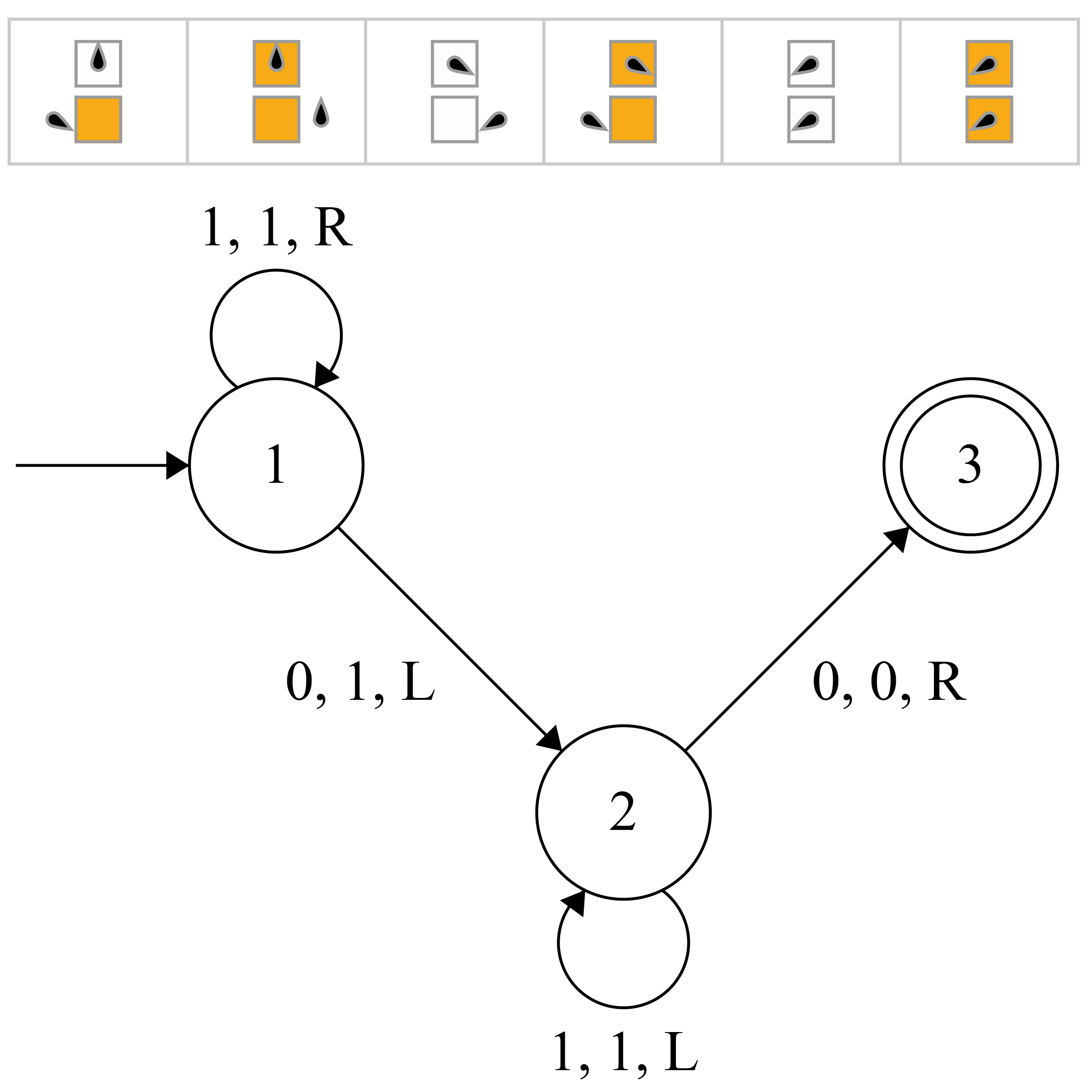}
  \caption{Rule plot and state transition diagram of $\mathcal{T}_1$. The arrow labels indicate \enquote{read, write, move}.}
  \label{fig:turingExample}
\end{subfigure}
\hfill
\begin{subfigure}{.45\textwidth}
  \centering
  \includegraphics[height=16em]{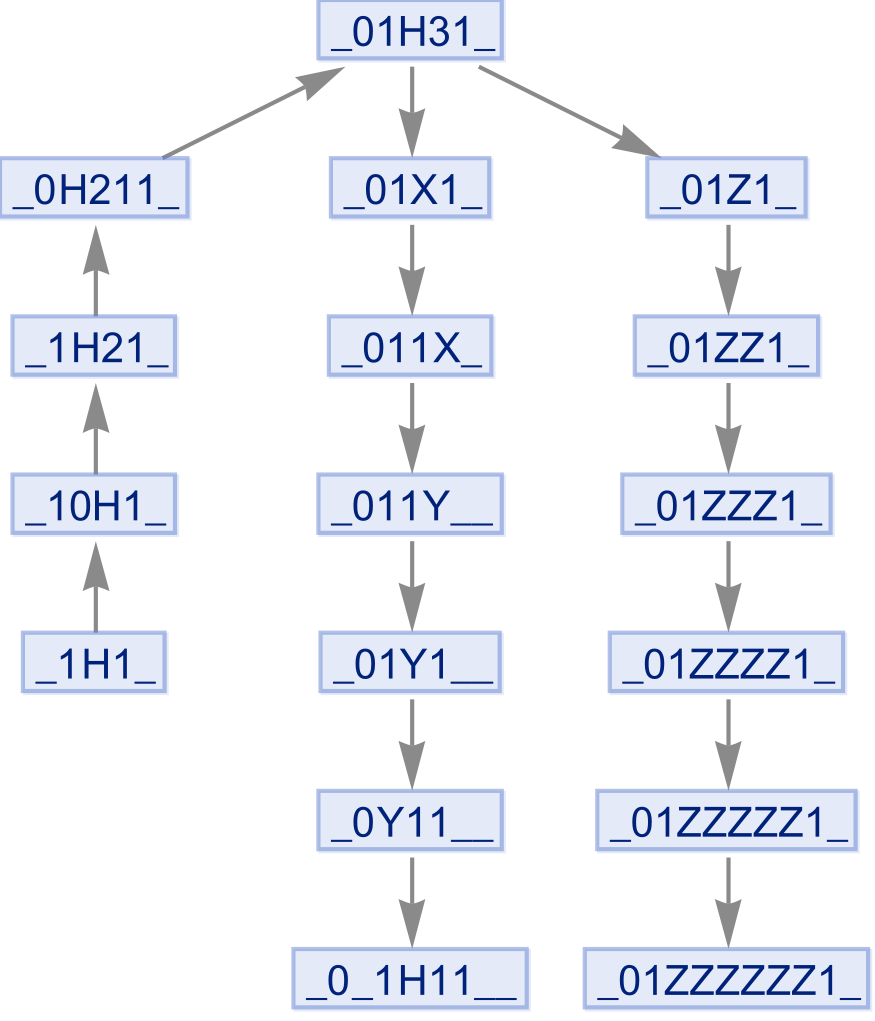}
\caption{States graph of the multiway system constructed from $\mathcal{T}_1$ (see below).}
\label{fig:trivialtHalterMultiwayGraph}
\end{subfigure}%
\caption{Two different representations of the same computational system: a Turing machine and a multiway system.}
\end{figure}

Now it is clear that a multiway system emulating some $T$-halter $\mathcal{T}$ has exactly one state for $T(n)$ generations when started with the initial condition $``\_ 1 H1 1^{n-1}\_"$ ($1^{n-1}$ denotes $n-1$ times the symbol $1$, the first $1$ is left of the head and the head starts in state 1). Now, we need a way to \enquote{enchain} this multiway system with itself. To achieve this, first add the rules \texttt{$``H f"\rightarrow ``X"$} for every halting state $f$ where $X$ is one fixed symbol not contained in $\mathcal{T}$'s alphabet. These additional rules will cause the multiway state in generation $T(n)+1$ to look like \texttt{$``\_ w_1 0 1 X 1^n 0 w_2 \_ "$}\footnote{If the head is left of the first digit, $``\_ w_1 0 X 1^{n+1} 0 w_2 \_"$ works analogously.} where $w_1$ and $w_2$ are arbitrary words that might be created as byproducts in the working of $\mathcal{T}$. This is due to the $T$-halter constraints depicted in figure \ref{fig:tHalterTape} or, rather, the $T$-halter constraints were chosen precisely to cause such a configuration of the tape. \par 
Adding the rules \texttt{$``X 1"\rightarrow ``1 X"$}, \texttt{$``X 0"\rightarrow ``Y \_ 0"$}, \texttt{$``1 Y"\rightarrow ``Y 1"$} and \texttt{$``0 Y 1"\rightarrow ``0 \_ 1 H 1"$} will cause exactly one state where the $X$ has \enquote{moved} one position to the right for $n$ generations (first rule), then add an underscore behind the $n$ ones (second rule), \enquote{move back to the left} using the $Y$ for $n+1$ generations (third rule) and finally add an underscore at the left side, replacing $Y$ by the starting state symbol $``H 1"$ of $\mathcal{T}$ (see figure \ref{fig:trivialtHalterMultiwayGraph}). Now, the whole process can start again because the new underscores ensure a \enquote{fresh} new tape for $\mathcal{T}$ which now contains, by the $T$-halter constraints, $n+1$ as the next input for $\mathcal{T}$ to continue with while everything outside the bounding underscores will be ignored. \par 
The resulting multiway system of this continued re-evaluation of $\mathcal{T}$ will run indefinitely, subsequently running instances of $\mathcal{T}$ with larger and larger values of $n$. Despite that, it still has only one state in all generations. In order to make the number of states increase exactly when $n$ increases, i.\,e. some instance of $\mathcal{T}$ finished working, we add the rules \texttt{$``0Y1"\rightarrow ``Z"$} and \texttt{$``Z"\rightarrow ``ZZ"$}. This way, the multiway states graph branches every time the system starts a new instance of $\mathcal{T}$ into a main branch where the evaluation of $\mathcal{T}$ continues and a diverging branch where the second rule just creates longer and longer strings of $Z$'s forever, constantly adding one new state to every generation. Thus, the number of divergent branches is always equal to $n-1$ and these branches grow constantly forever, causing the desired behavior as shown in figure \ref{fig:trivialtHalterMultiwayGraph}. \par
As it takes a \enquote{preparation time} of $p(n)=2(n+1)+1$ steps (the head moves over $n+1$ symbols\footnote{If the head starts at the left of the first digit instead, the formula is $p(n)=2(n+2)$.}, including the new 1) before the $n+1$-th iteration of $\mathcal{T}$ can start after the $n$-th iteration is done, there will be $n$ states for $T(n) + p(n)$ steps in the multiway system constructed above, before the number of states increases by one. Let us generally investigate the sequences obtained this way:

\begin{definition}\label{def:nOccursFtimes}
Let $f:\N_+\rightarrow \N_+$ be a function. The sequence \enquote{$n$ occurs $f(n)$ times} is defined by $A_f(\sum_{k=1}^n f(k))=A_f(m+ \sum_{k=1}^n f(k))=n$ for all $n,m\in \N$ with $n\geq 1 \land m < f(n+1)$. 
\end{definition}

\begin{definition}\label{def:linearInterpolation}
Let $f:\N_+\rightarrow \N_+$ be a function and $S\subseteq \N_+$ an infinite set. The \enquote{linear interpolation of $f$ over $S$}, denoted $L_S(f)$, is defined as the polygonal chain starting at $(0,0)$ and passing through all points $(n,f(n)), n\in S$ ordered by $n$.   
\end{definition}

\begin{figure}[ht]
\centering
\begin{subfigure}{.45\textwidth}
  \centering
  \includegraphics[width=\linewidth]{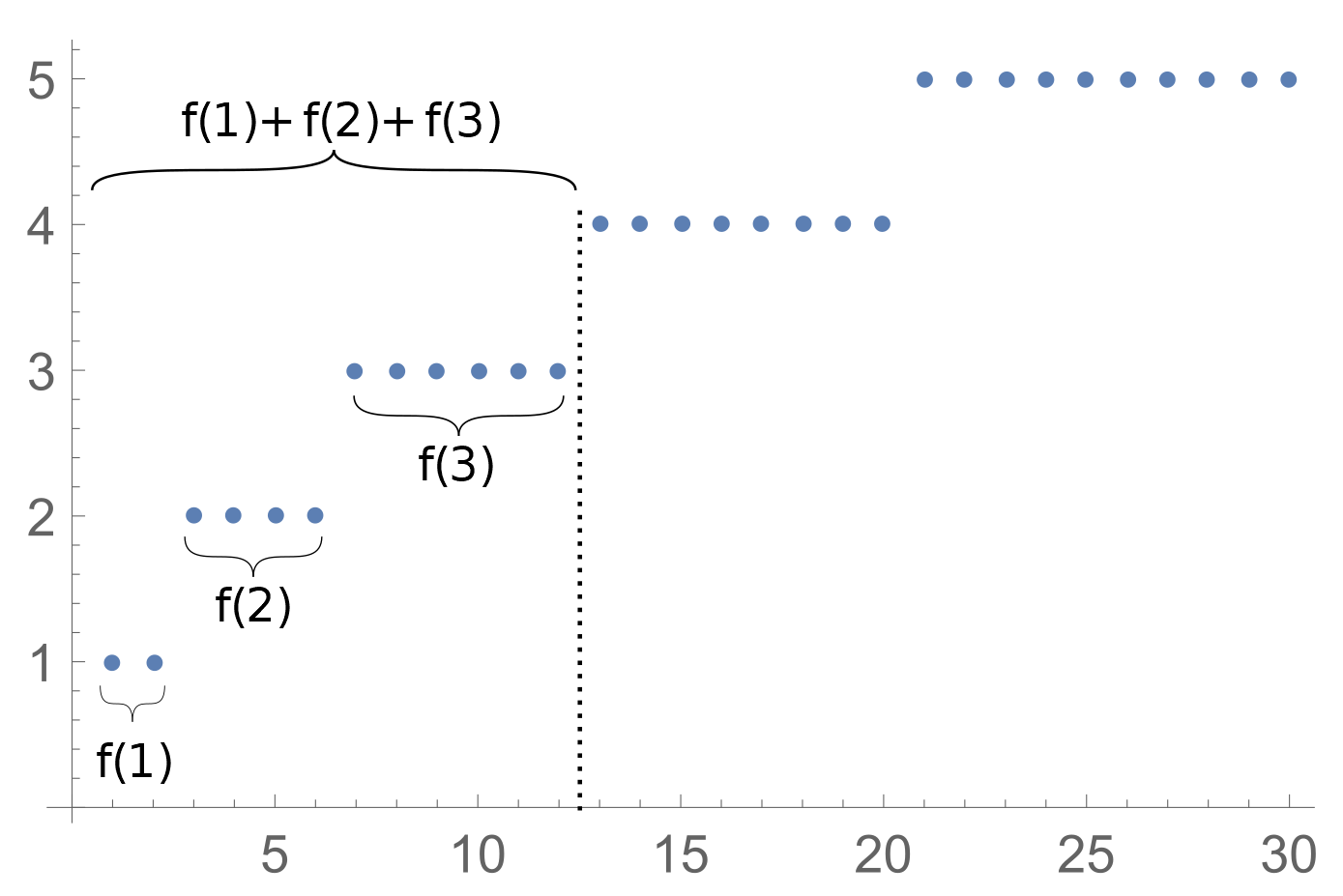}
  \caption{Example for definition \ref{def:nOccursFtimes}: the sequence \enquote{$n$ occurs $f(n)=2n$ times} ($A_{2n}$).}
  \label{fig:nOccurs2nTimes}
\end{subfigure}
\hfill
\begin{subfigure}{.45\textwidth}
  \centering
  \includegraphics[width=\linewidth]{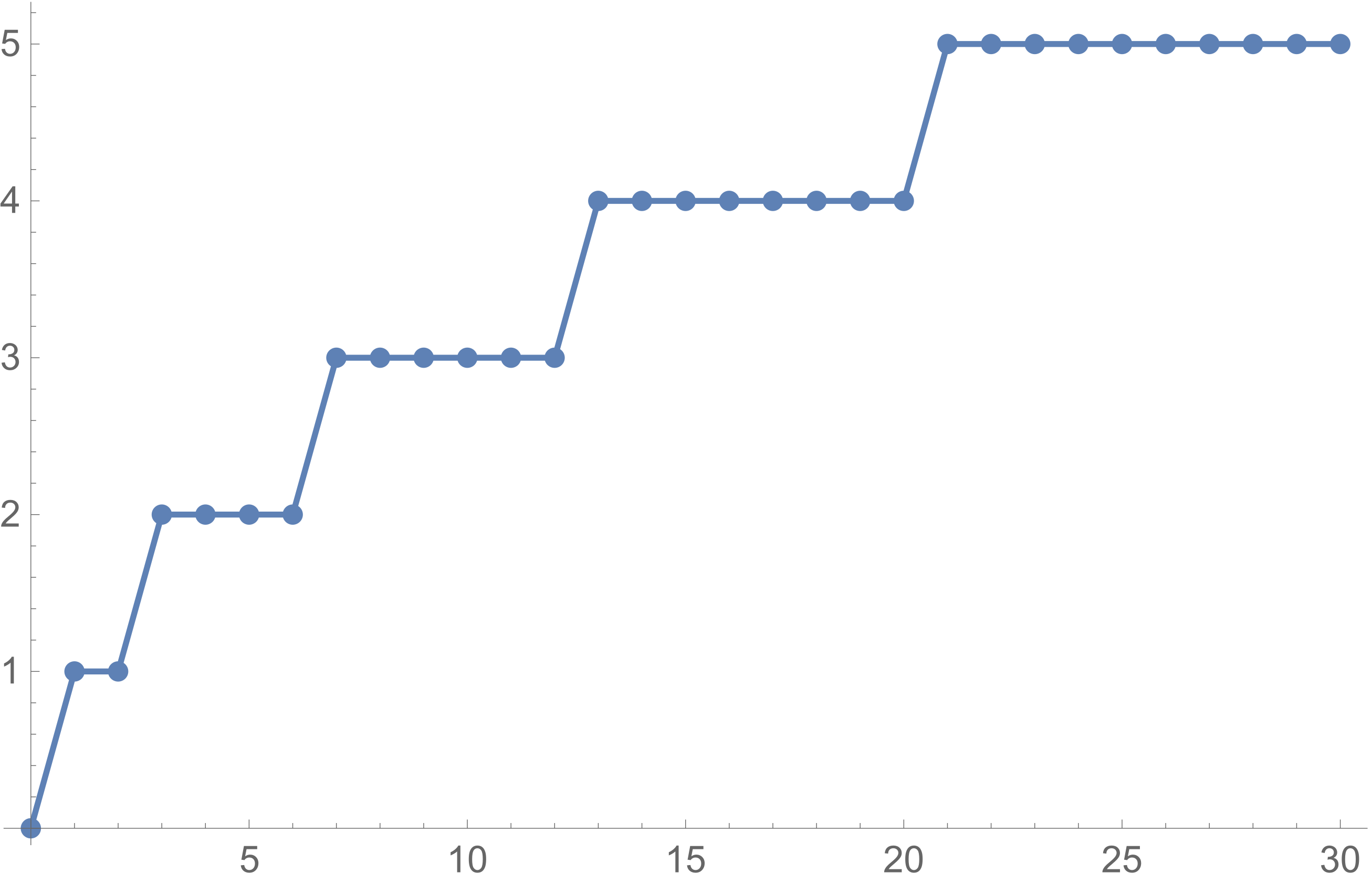}
  \label{fig:nOccurs2nTimesLinearInterpolation}
    \caption{Example for definition \ref{def:linearInterpolation}: the linear interpolation of $A_{2n}$ over $\N_+$, now a continuous function from $\R_{\geq 0}$ to $\R_{\geq 0}$.}
\end{subfigure}%
\caption{Plots for illustrating definitions \ref{def:nOccursFtimes} and \ref{def:linearInterpolation}.}
\end{figure}

Since definition \ref{def:nOccursFtimes} requires $f$ to be always greater than zero, every natural number can be represented as some sum over consecutive values of $f$ plus a remainder and, as figure \ref{fig:nOccurs2nTimes} shows, this definition indeed matches the informal description of \enquote{$n$ occurs $f(n)$ times}. Notice as well that the linear interpolation, despite being defined as a curve in $\R^2$, can be regarded as a continuous function from $\R_{\geq 0}$ to $\R_{\geq 0}$ because for all $n\in S$, the function to be interpolated assigns precisely one $y$-value and since $S$ is an infinite subset of $\N$, the linear interpolation function is defined everywhere on $\R_{\geq 0}$.\par 
Now, to express some sequence $A_f$ explicitly, define the set of increase-indices of $A_f$ as $I(A_f):=\{ n \in \N_+ \mid A_f(n-1) < A_f(n) \}$. It follows that $L_{I(A_f)}(A_f)$ will always be strictly monotonically increasing and unbounded. Therefore, its inverse function $L_{I(A_f)}(A_f)^{-1}$ exists and we can formulate the following lemma:

\begin{lemma}\label{lemma:inverseSummatory}
For some function $f:\N_+\rightarrow \N_+$, we have $L_{I(A_f)}(A_f)(x)=L_{\N_+}(\sum_{k=1}^n f(k))^{-1}(x)$.
\end{lemma}
\begin{proof}
For readability, let $T_\sigma$ be the function $\sum_{K=1}^n T(k)$, $\phi(x):=L_{I(A_f)}(A_f)(x)$ and $\psi(x):=L_{\N_+}(\sum_{k=1}^n f(k))(x)$. By definition \ref{def:linearInterpolation}, the linear interpolation of a function equals that function on the interpolation set, so 
\begin{equation}
    \forall n\in \N_+: (\phi \circ \psi)(n) = \phi(\sum_{k=1}^n f(k)) = n
\end{equation}
For values $x\in (n,n+1), n\in \N$, the linear interpolation gives
\begin{align}\label{eq:psiXinvSummator}
\psi(x)&=\frac{\Delta y}{\Delta x}(x-n) + \psi(n) =  \frac{\psi(n+1)-\psi(n)}{n+1-n}(x-n)+\psi(n) \notag \\ &=(\psi(n+1)-\psi(n))(x-n)+\psi(n).
\end{align}
Letting $y=\psi(x)$, we know that 
\begin{equation}
    (\phi\circ\psi)(x)=\phi(y)=\frac{\phi(y_2)-\phi(y_1)}{y_2-y_1}(y-y_1)+\phi(y_1)
\end{equation} 
for some $y_1,y_2\in I(A_f)$ where $y_1<y<y_2$ and $y_1,y_2$ are the values in $I(A_f)$ closest to $y$. Since $\psi$ is strictly monotonically increasing, $y_1$ and $y_2$ must be given by $\psi(n)$ and $\psi(n+1)$ respectively. Thus, 
\begin{align}{c|r}
    (\phi\circ\psi)(x)&=\frac{\phi(\psi(n+1))-\phi(\psi(n))}{\psi(n+1)-\psi(n)}(\psi(x)-\psi(n))+\phi(\psi(n)) & \\
    &=\frac{n+1-n}{\psi(n+1)-\psi(n)}(\psi(x)-\psi(n))+n & \notag \\
    &= \frac{(\psi(n+1)-\psi(n))(x-n)+\psi(n) - \psi(n)}{\psi(n+1)-\psi(n)} + n & \text{by equation \ref{eq:psiXinvSummator}} \notag \\
    &=x-n+n=x.\notag  &
\end{align}
So $\phi$ is a left-inverse of $\psi$ on $\R_{\geq 0}$. Analogously, it can be shown that $\phi$ is also a right-inverse of $\psi$, so indeed, $L_{I(A_f)}(A_f)(x)$ and $L_{\N_+}(\sum_{k=1}^n f(k))(x)$ are inverse functions. 
\end{proof}

Putting it all together, we conclude from the previous Turing machine investigation that for every halting function $T\in \Omega(n)$, there is a multiway system which has the growth function $g_M(n)=A_{T+p}(n)$ for some $p\in \Theta(n)$. Additionally, as $p$ \enquote{delays} the growth function even more\footnote{Since $p\in \Theta(n)$, $g_M$ will even become strictly less than $L_{\N_+}(T_\sigma)^{-1}$ very soon. In most practical cases, $g_M$ is much lower.}, i.\,e. $\forall n\in \N_+: A_{T+p}(n)\leq A_T(n)$, it follows that 
\begin{equation} 
L_{I(g_M)}(g_M)(n) \leq L_{I(A_T)}(A_T)(n) = L_{\N_+}(T_\sigma)^{-1}(x) \Rightarrow g_M(n) \in \mathcal{O}\left(L_{\N_+}(\sum_{k=1}^n T(k))^{-1}\right)
\end{equation}
by lemma \ref{lemma:inverseSummatory}, which proves theorem \ref{theorem:slowness}.

\subsection{Applications of Theorem \ref{theorem:slowness}}\label{section:applicationSlowness}

Computing linear interpolations and their inverse functions seems hard to do analytically because in most cases, there are no elementary closed-form expressions describing them. Therefore, it might seem difficult to actually apply theorem \ref{theorem:slowness}. However, since we are only interested in growth rates, we can use approximations to make calculations much more easy. 

\begin{lemma}\label{lemma:analyticalInvert}
If $f:\N_+\rightarrow \N_+$ is strictly increasing, $g:\R_{\geq 0}\rightarrow \R_{\geq 0}$ is continuous and bijective and $\forall n\in\N_+: g(n)=f(n)$, then $L_{\N_+}(f)^{-1}\in \Theta(g^{-1})$. 
\end{lemma}
\begin{proof}
Because $g$ is continuous and takes the same values as $f$ for natural arguments, we know that $\forall n\in \N_+: \forall x\in (n,n+1): f(n)\leq g(x)\leq f(n+1)$. Using the fact that the linear interpolation equals the function for natural arguments, this equation becomes $L_{\N_+}(f)(n)\leq g(x)\leq L_{\N_+}(f)(n+1)$. This implies $L_{\N_+}(f)^{-1}(y_1)\leq g^{-1}(y) \leq L_{\N_+}(f)^{-1}(y_2)$ for values $y\in (y_1,y_2)$ where $y_1=f(n)$ and $y_2=f(n+1)$. Expanding out gives
\begin{equation}
    L_{\N_+}(f)^{-1}(f(n))\leq g^{-1}(y) \leq L_{\N_+}(f)^{-1}(f(n+1)) \;\Rightarrow\; n\leq g^{-1}(y)\leq n+1
\end{equation}
which means that the difference of $g^{-1}(y)$ and $L_{\N_+}(f)^{-1}(y)$ is always bounded by 1. Therefore, $g^{-1}\in\Theta(L_{\N_+}(f)^{-1})$.
\end{proof}

Whenever it is possible to express some halting function in a closed form (e.\,g. $T:\N_+\rightarrow \N_+, n\mapsto 2n^2+3n$) which could also describe a bijective function on $\R_{\geq 0}$ (like $f(x)=2x^2+3x$), we can use the lemma above to simplify calculations: Since $f$ is monotonically increasing and the linear interpolation equals the summatory function for natural arguments, we have
\begin{equation}\label{eq:integralInequality}
    \sum_{k=0}^{\floor x} f(k) \leq \int_0^x f(t)\d t \leq \sum_{k=1}^{\ceil x} f(k) \;\Rightarrow\; L_{\N_+}(T_\sigma)(\floor x) \leq F(x) \leq L_{\N_+}(T_\sigma)(\ceil x).
\end{equation}
Therefore, lemma \ref{lemma:analyticalInvert} tells us that we can approximate the inverse of the summatory function used in theorem \ref{theorem:slowness} just by computing the inverse of the integral of $f$. Especially in the case of logarithms or exponential functions, solving integrals is much easier than computing sums, so this lemma can be very useful. \par 
To demonstrate this and assist the proof of theorem \ref{theorem:classesPartition}, let us imagine we wanted to construct a multiway system of logarithmic growth rate. We can approach this problem by designing a $T$-halter for some $T_{\text{exp}}\in \Theta(2^n)$ and implementing the construction described in section \ref{section:slownessProof} to get a multiway system with the inverse growth rate. As an example, take the Turing machine $\mathcal{T}_{\text{exp}}$ shown in figure \ref{fig:expTm}. It is started in state 1 which simply moves the head to the right end of the word on the tape and changes to state 2. In this state, the head moves left again, replacing 2's by 1's until it encounters a 1, which it changes to a 2 and returns to state 1, repeating the process. It is easy to see that this is precisely the process of incrementing a binary number where 1 corresponds to a zero and 2 to a one. The process is repeated until the head moves to the left of the word, which, by then, consists only of 1's since the previous string was the symbol $2$ repeated $n$ times. When the head encounters the first blank symbol on the left, it writes one more 1 to satisfy the $T$-halter constraint of incrementing the unary number, and then halts. The process is shown in figure \ref{fig:expTm} for $n=3$.

\begin{figure}[ht]
\centering
\begin{subfigure}{\textwidth}
  \centering
  \includegraphics[width=.7\linewidth]{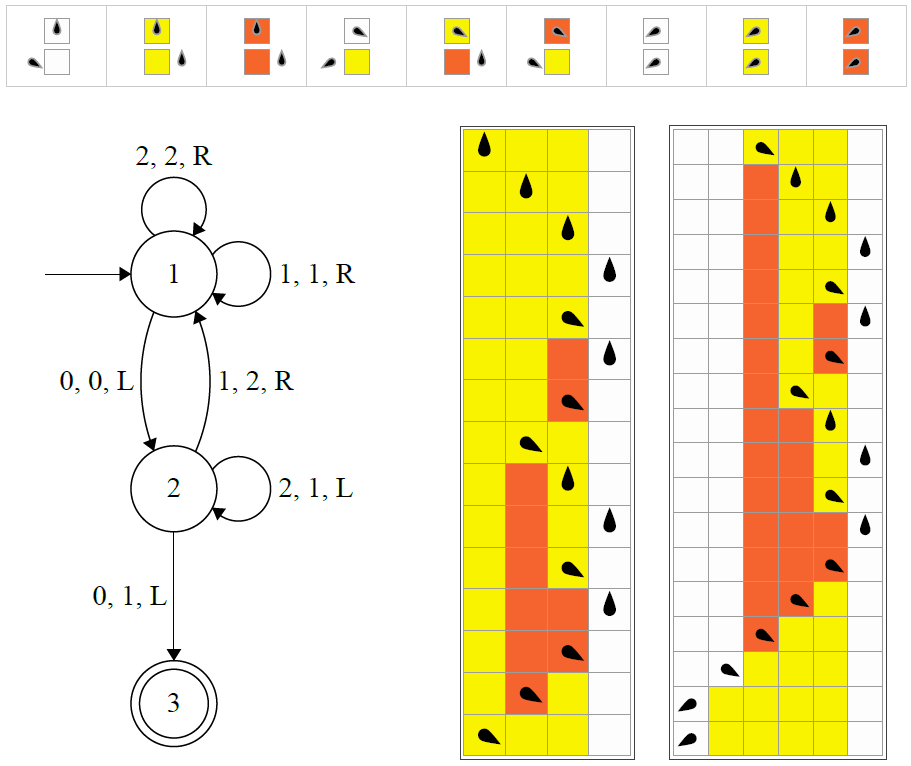}
\end{subfigure}
\caption{The rule plot, state transition diagram and one example evolution (starting from three ones) for the Turing machine $\mathcal{T}_{\text{exp}}$.}
\label{fig:expTm}
\end{figure}
\newpage 
Consider the action of the machine when started at the right of some binary word of 1's and 2's in state 1: The head moves $b$ symbols to the left until it encounters the rightmost 1, and then $b$ symbols back after changing it. Since there are $2^{n-b}$ binary words of length $n$ where the rightmost $1$ is at position $b$, the Turing machine takes 
\begin{align}
    \sum_{k=1}^n 2k\cdot 2^{n-k}&=2\sum_{k=1}^n (n-k)2^k=2\left(n\sum_{k=1}^n 2^k - \sum_{k=1}^n k2^k \right)& \notag \\
    &= 2\left( n(2^{n+1}-1) - \frac{2-(n+1)2^{n+1}+n2^{n+2}}{(2-1)^2}\right) & \text{geometric series and \cite{mathWorld2}} \notag \\
    &= 2(n2^{n+1}-n-2+n2^{n+1}+2^{n+1}-n2^{n+2}) & \notag \\
    &= 2(-n-2+2^{n+1}) = 2^{n+2}-2n-4 & 
\end{align} 
steps before the head is at the right of $1^{n}$ and in state 1. Since the machine takes $n+2$ steps to move to the left again, write the new 1 and halt, as well as taking $n$ steps to move the head to the right in the first place, the total number of states, including starting and halting state, simplifies to
\begin{equation}
    T_{\text{exp}}(n)=2n+3+2^{n+2}-2n-4=2^{n+2}-1.
\end{equation}
In combination with lemma \ref{lemma:analyticalInvert}, another strategy for simplifying calculations is to give easily computable bounds for $T$. In this case, we use the fact that $2^{n+1} < 2^{n+2}-1 < 2^{n+2}$ for all $n\in \N$ to obtain $2^{x+1} < L_{\N_+}(T_{\text{exp}})(x) < 2^{x+2}$ for all $x\in \R_{\geq 0}$. Letting $l(x):=L_{\N_+}(T_{\text{exp}})(x)$ for readability, this becomes
\begin{align}
    2^{x+1} < l(x) < 2^{x+2} &\iff \int_0^x 2^{t+1} \d t < \int_0^x l(t) \d t < \int_0^x 2^{t+2} \d t \notag \\
    &\iff \frac{2}{\ln(2)}(2^x - 1) < \int_0^x l(t) \d t < \frac{4}{\ln(2)}(2^x - 1). 
\end{align}
Since the inverse of $x\mapsto \frac{a}{\ln(2)}(2^x - 1)$ is $y\mapsto \log_2(y\frac{\ln(2)}{a} + 1)$, and $f(x)<g(x)\iff f^{-1}(x)>g^{-1}(x)$, the equation is equivalent to 
\begin{equation}
    \log_2\left(y\frac{\ln(2)}{2}+1\right) > \left(\int_0^x l(t) \d t\right)^{-1} > \log_2\left(y\frac{\ln(2)}{4}+1\right).
\end{equation}
Now, notice that
\begin{equation}
    \log_2(y+c)=\log_2(y)+\log_2\left(1+\frac{c}{y}\right)\overset{x\rightarrow \infty}{\longrightarrow}\log_2(y) \quad \text{(since $1+\frac{c}{y} \overset{x\rightarrow \infty}{\longrightarrow} 1$)} 
\end{equation}
and
\begin{equation}
    \log_2(yc)=\log_2(y)+\log_2(c)\quad \land \quad \log_2(y)+\log_2(c) \overset{x\rightarrow \infty}{\longrightarrow}\log_2(y).
\end{equation}
From this, we know $(\int_0^x l(t) \d t)^{-1}\in \Theta(\log_2(x))$ because the upper and lower bound asymptotically equal $\log_2(x)$. By lemma \ref{lemma:analyticalInvert} and equation \ref{eq:integralInequality}, $(\int_0^x l(t) \d t)^{-1}$ is also in $\Theta(L_{I(A_{T_{\text{exp}}})}(A_{T_{\text{exp}}}))$ and we can conclude that $A_{T_{\text{exp}}} \sim_\Theta \log_2(x)$. Simulating the multiway system and measuring the growth function empirically supports this as figure \ref{fig:expSysGrowthRate} shows. In future examples, most steps of the argumentation presented here can be shortened. However, this method of estimation does not work in all cases because the inverse bounds might not accurate enough to be asymptotically equal\footnote{It is not true in general that $f \in \Theta(g)$ implies $f^{-1}\in \Theta(g^{-1})$. As a counter-example, consider $f(x)=\ln(x)$ and $g(x)=2\ln(x)$.}. 

\begin{figure}[ht]
\centering
\begin{subfigure}{\textwidth}
  \centering
  \includegraphics[width=.65\linewidth]{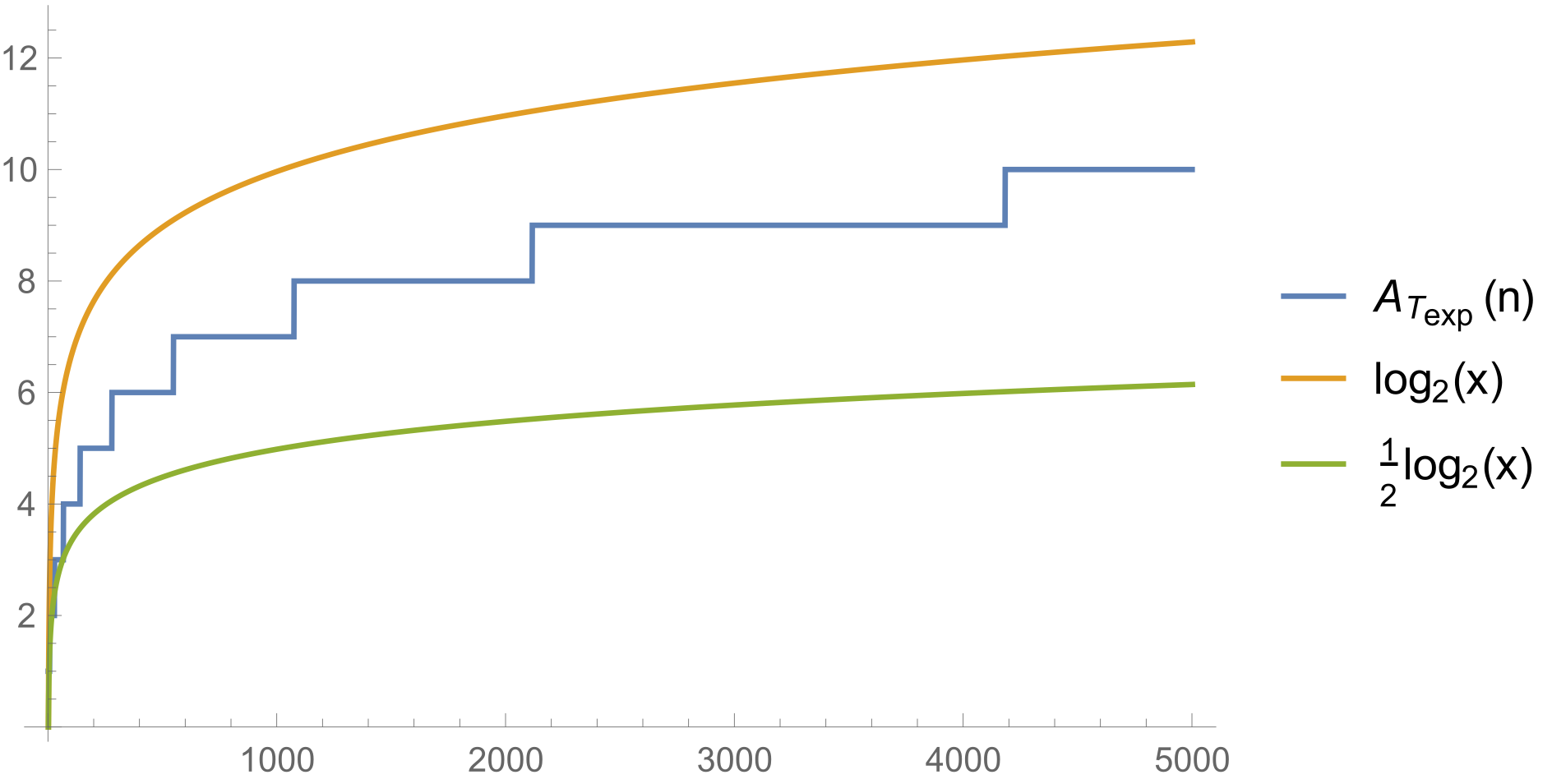}
\end{subfigure}
\caption{The growth function of the multiway system emulating $\mathcal{T}_{exp}$ is bounded by $\log_2(x)$ and $\frac{1}{2}\log_2(x)$, demonstrating that it is in $\Theta(\log_2(x))$.}
\label{fig:expSysGrowthRate}
\end{figure}

\subsection{Implications of Theorem \ref{theorem:slowness}}

What we have seen in the above example is just a simple demonstration of the power of theorem \ref{theorem:slowness}. Besides from helping us later to prove theorem \ref{theorem:classesPartition}, it tells us a lot about the abstract structure of multiway growth function, their \enquote{growth spectrum}. By providing the following two corollaries, theorem \ref{theorem:slowness} gives us knowledge about what this spectrum of possible growth rates looks like, i.\,e. which kinds of growth rates are possible and which kinds are not. In addition to that, it establishes connections between multiway growth functions and other classes of functions, namely computable functions and primitive recursive functions.

\begin{corollary}\label{corollary:multiwaySlowerInverseComputable}
For every computable function $f:\N\rightarrow \N,f\in \omega(1)$, there is a multiway system with growth rate $(a,b)$ such that $a,b\in \mathcal{O}(f^{-1})$ for an asymptotic inverse\footnote{An asymptotic inverse of a function $f$ is some function $f^{-1}$ satisfying $(f^{-1})^{-1} \in \Theta(f)$. Such functions are helpful for describing $f$ when it is not invertible in general.} $f^{-1}$.
\end{corollary}
\begin{proof}
Since $f$ is computable, there exists some Turing machine computing $f(n)$ when given $n$. If we require the machine to read and write in- and output in unary coding, computing $f(n)$ must take $T(n)\geq f(n)$ steps simply because writing the result takes that long. Now, let $g:\R_{\geq 0}\rightarrow \R_{\geq 0}$ be a bijective tight lower bound of $T$. As $T\in \omega(1)$, $T$ is unbounded so $g$ always exists. From $g(x)\leq T(x)$, it follows that $g^{-1}(x)\geq T^{-1}(x)$ and by theorem \ref{theorem:slowness}, there is a multiway system for which the growth function has tight bounds $a,b\in \mathcal{O}(g^{-1})\Rightarrow a,b\in \mathcal{O}(f^{-1})$ for some asymptotic inverse of $f$.
\end{proof}

\begin{corollary}\label{corollary:multiwaySlowerComputable}
For every computable function $f:\N\rightarrow \N, f\in \omega(1)$, there is a multiway system with growth rate $(a,b)$ such that $a,b\in o(L_{\N_+}(f))$.
\end{corollary}
\begin{proof}
The function $L_{\N_+}(\overline{f})$ (for the upper bounding sequence $\overline{f}$ from definition \ref{def:tightBounds}) is always greater than or equal to $f$, asymptotically equal to $f$ and computable (because equal to $f$) on the set of increase-indices $I(f)$. Therefore, $\overline{f}^{-1}$ is a computable function on $\N_+$ and so is $g:x\mapsto \overline{f}^{-1}(x^2)$. Using corollary \ref{corollary:multiwaySlowerInverseComputable}, this gives us a way to construct a multiway system with a growth rate in $\mathcal{O}(g^{-1})=\mathcal{O}(\sqrt{L_{\N_+}(f)})$ which is definitely in $o(L_{\N_+}(f))$.
\end{proof}

This quite remarkable fact also shows that for every multiway system growing faster than a bounded function, a more slowly growing multiway system exists because the growth function of every multiway system is obviously computable. We might therefore say that multiway systems can grow arbitrarily slowly, i.\,e. the set of regular multiway systems excluding constant and finite systems is \enquote{open} in some sense. Remember however that they cannot grow arbitrarily quickly as shown in lemma \ref{lemma:expBound}.

\section{Computational Capabilities of Growth Functions}\label{section:computationalCapabilities}

After marking out the boundaries of the space of possible growth rates, we shall investigate its underlying structure. First of all, we will see that it contains no \enquote{holes}, i.\,e. all of the multiway growth classes defined in section \ref{section:growthClassesDef} (except $C_{\text{supexp}}$ which we have already shown to be empty and just defined for completeness) are non-empty and, furthermore, contain infinitely many systems. In addition, we will have some insights into which functions are \enquote{multiway-growth-computable} and \enquote{multiway-growth-approximable}. We say, a function $f:\N_+\rightarrow \N_+$ is multiway-growth-computable if there is a multiway system $M$ such that $\forall n\in \N_+: g_M(n)=f(n)$ and we call a function $f:\R_{\geq 0}\rightarrow \R_{\geq 0}$ multiway-growth-approximable if there is a multiway system $M$ such that $f \sim_\Theta L_{\N}(g)$. \par 
First, we will define two operations, \enquote{multiway addition} and \enquote{multiway multiplication} which will enable us to combine systems into more complex ones of which the growth function is computable immediately from the growth functions of the parts. These two fairly simple operations will be sufficient for demonstrating that multiway growth functions are interesting from an algebraic point of view as well as regarding questions of their computational capabilities (see section \ref{section:consequencesComputability}). Still, some basic multiway systems have to be constructed without using these operations as the building blocks of further systems. Combining the multiway operations and specifically constructed systems will then yield the following theorem and several other interesting results: 

\begin{theorem}\label{theorem:classesPartition}
The classes $C_{\text{pol}}, C_{\text{int}}, C_{\text{exp}}, C_{\text{invpol}}, C_{\text{invint}}, C_{\text{invexp}}, C_{\text{invsupexp}}, C_{\text{fin}}$ and $C_{\text{bnd}}$ partition the set of regular multiway systems into infinite subsets.
\end{theorem}

\subsection{Arithmetic-like operations on multiway systems}\label{sec:multiwayArithmetic}

Let $M_1=(R_1,s_1,\Sigma_1),M_2=(R_2,s_2,\Sigma_2)$ and $M_3=(R_3,s_3,\Sigma_3)$ be multiway systems. Additionally, let $X$ be a unique (equal for all multiway systems) symbol not included in any multiway systems alphabet. Now, we define the \enquote{sum system} by $M_1 \oplus M_2 = (R_1\cup R_2\cup \{X\rightarrow s_i\mid s_i\in S_2(M_1)\cup S_2(M_2)\}, ``X", \Sigma_1\cup \Sigma_2)$ where $S_2(M)$ is the state-set of $M$ in generation 2, i.\,e. all nodes with distance 1 to the initial state in the respective states graphs. The \enquote{product system} of $M_1$ and $M_2$ is now defined as $M_1\odot M_2 = (R_1\cup R_2, s_1 s_2, \Sigma_1 \cup \Sigma_2)$ where $s_1 s_2$ denotes the concatenation of $s_1$ and $s_2$.\par 

\begin{figure}[ht]
\centering
\begin{subfigure}{.45\textwidth}
  \centering
  \includegraphics[height=14em]{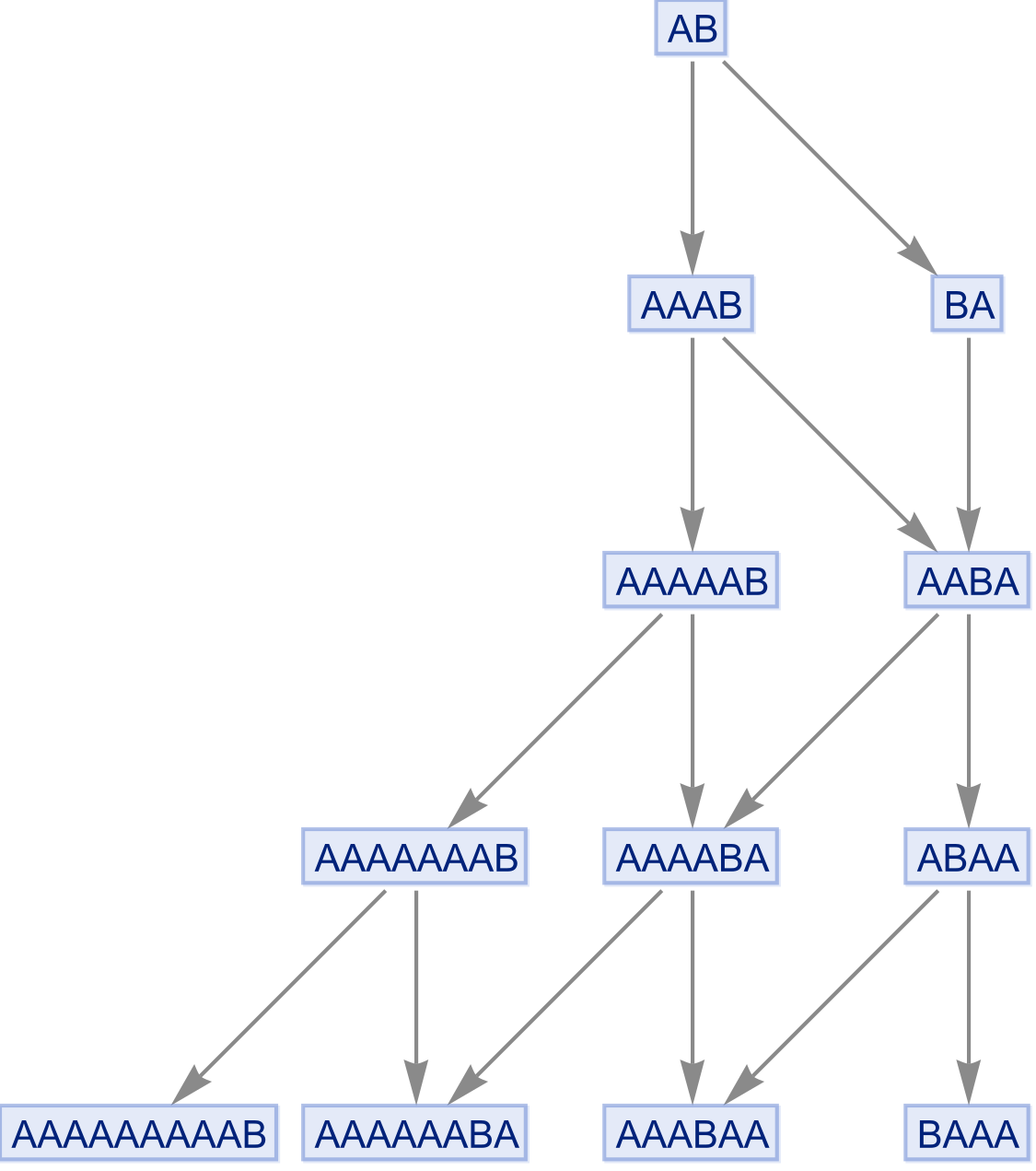}
\end{subfigure}
\begin{subfigure}{.45\textwidth}
  \centering
  \includegraphics[height=14em]{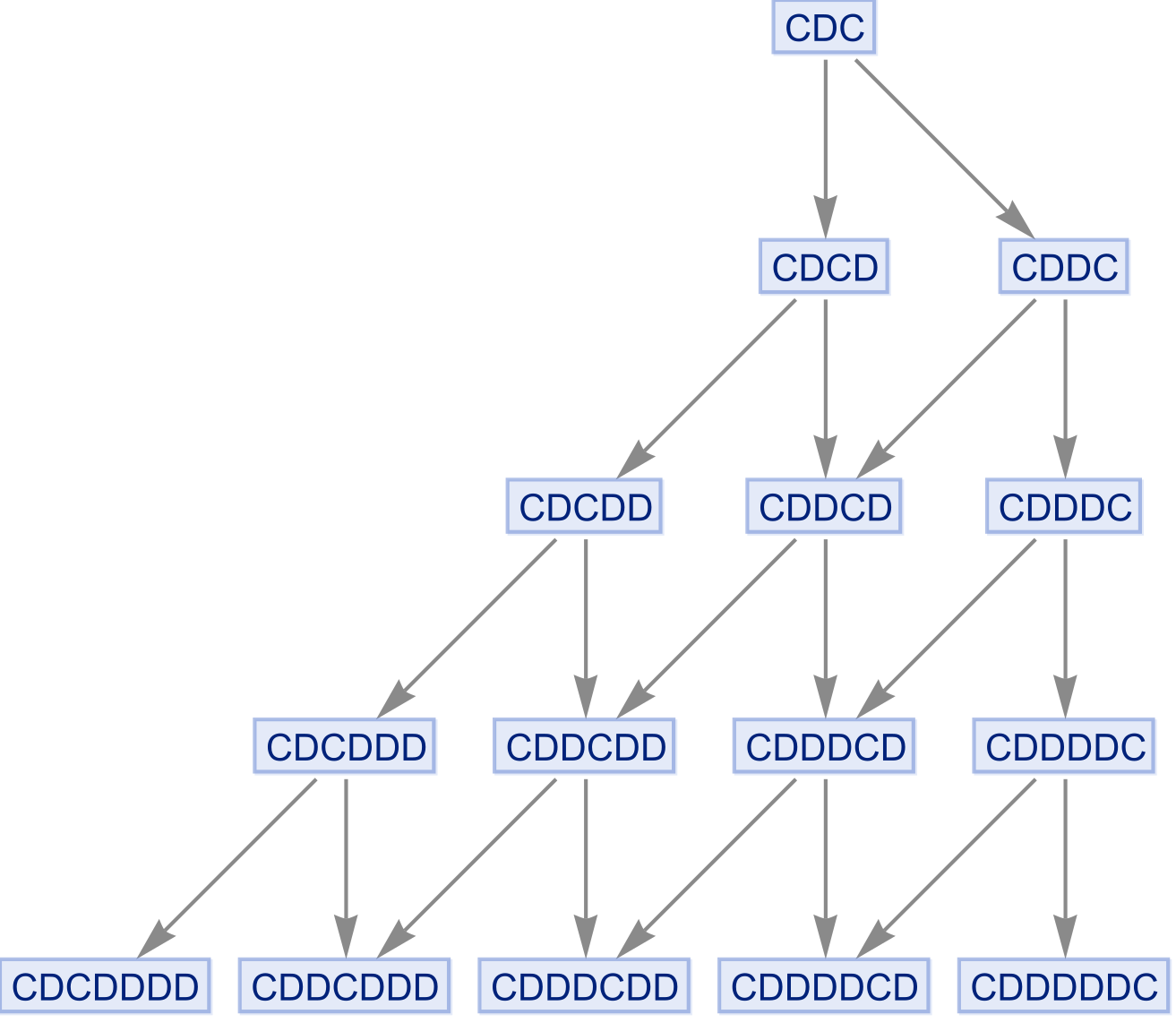}
\end{subfigure}\\
\begin{subfigure}{.75\textwidth}
    \centering
    \includegraphics[height=16em]{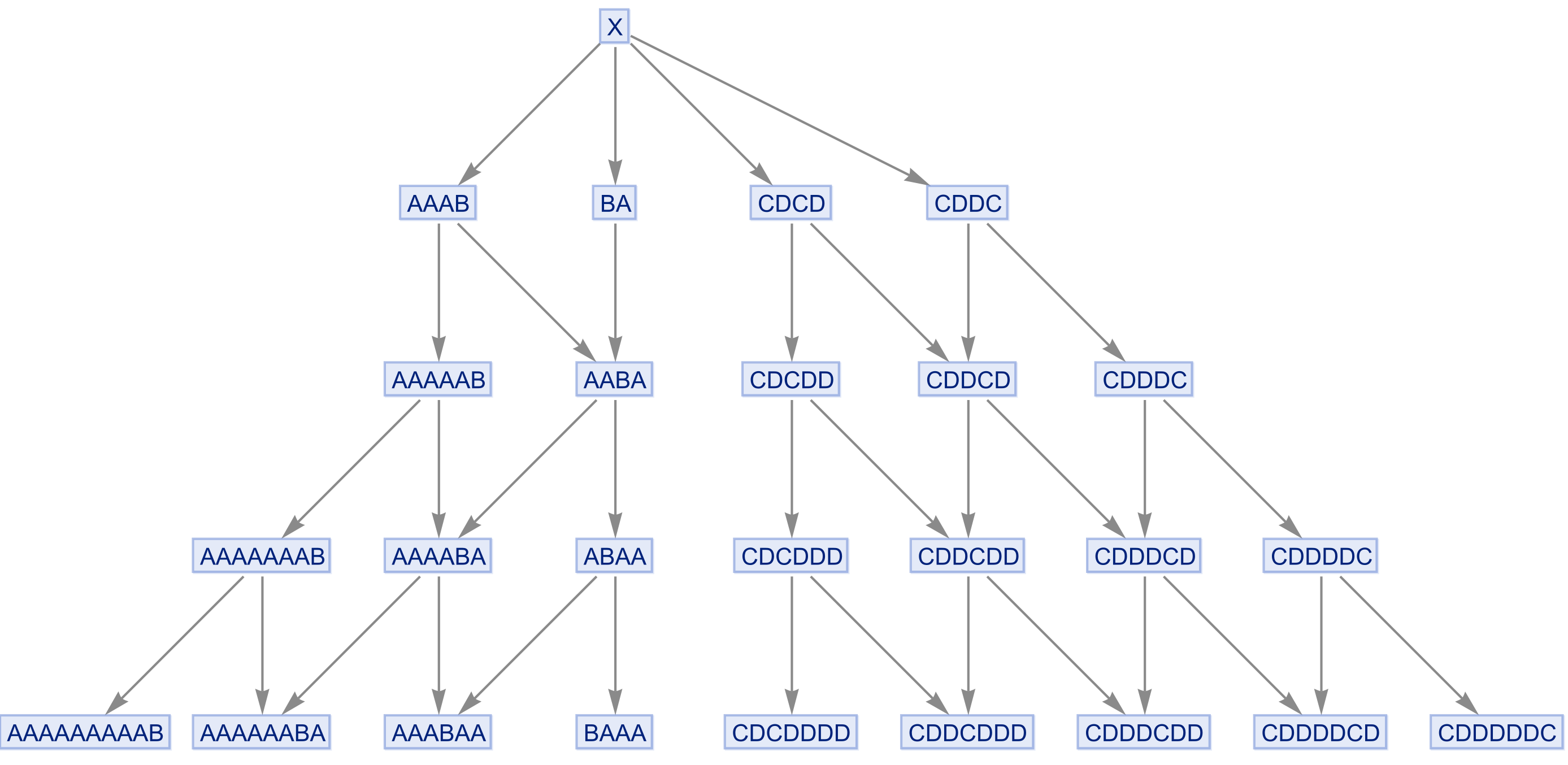}
\end{subfigure}
\caption{States graphs of rule independent multiway systems and their sum system. The systems used are $(\{``AB"\rightarrow``BA",``B"\rightarrow ``AAB"\},``AB",\{A,B\})$ and $(\{``CD"\rightarrow``CDD",``C"\rightarrow``CD"\},``CDC",\{C,D\})$.}
\label{fig:multiwayAddition}
\end{figure}

\begin{figure}[!h]
\centering
\begin{subfigure}{\textwidth}
  \centering
  \includegraphics[width=\linewidth]{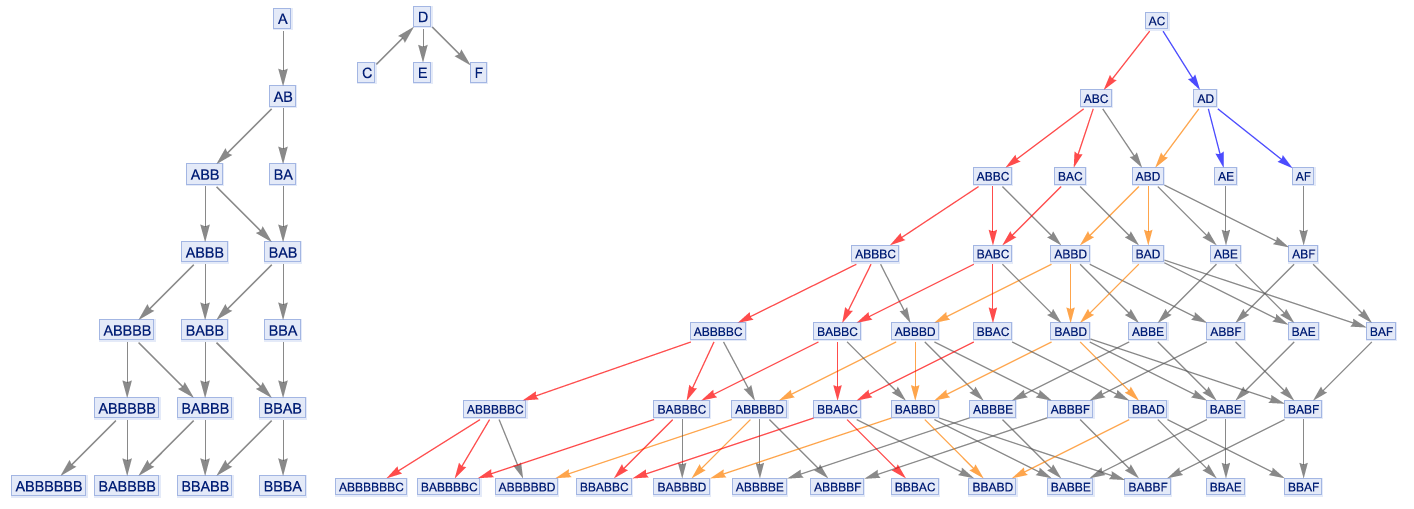}
\end{subfigure}
\caption{States graphs of $(\{``A"\rightarrow``AB",``AB"\rightarrow``BA"\},``A",\{A,B\})$, $(\{``C"\rightarrow``D",``D"\rightarrow``E",``D"\rightarrow``F"\},``C",\{C,D,E,F\})$ and their product system. Instances of the first system are highlighted in red and orange in the product system's graph.}
\label{fig:multiwayProduct}
\end{figure}

To calculate the growth functions of systems obtained by these operations, we shall require the parts $M_1$ and $M_2$ to be \enquote{rule independent} meaning that their rules do not interfere with each other. Formally, rule independence can be defined as the property that the states graph of $M_1$ is isomorphic to the states graph of $(R_1\cup R_2, s_1, \Sigma_1\cup \Sigma_2)$, that is $M_1$ with all rules of $M_2$ added, and vice versa\footnote{This works because if adding all the rules of $M_1$ to $M_2$ does not change its behavior, then these rules will not influence $M_2$'s states even if the states of $M_1$ get appended to them.}. This property can always be achieved by requiring the underlying alphabets to be disjoint as in this case it will be impossible that a given string matches rules from $R_1$ and $R_2$ at the same time.\par  
If we recall the definition of growth function $g_M(n)$ as the number of nodes to which the shortest path from the initial state has length $n$, it is easy to see that the growth function of $M_1\oplus M_2$ is one in the first iteration as $``X"$ is the only state. In further iterations, we can imagine a path of length $n$ simply as a path of length 1 entering either the states graph of $M_1$ or $M_2$, followed by a path of length $n-1$ originating at some state in the second layer of the chosen subgraph, as if the other graph was not there. This works because in the entire evolution, the initial state is the only one containing an $X$ so for all other states, the usual rules will apply and the rules replacing $X$ will have no effect. This concludes that the growth function of the sum system is precisely 
\begin{equation}\label{eq:sumSystem}
    g_{M_1\oplus M_2}(n)=g_{M_1}(n)+g_{M_2}(n) \quad\text{with}\quad g_{M_1\oplus M_2}(1):=1.
\end{equation}
For the product system, remember that in the states graph of some system $M_1$, two nodes $u,v$ are connected by an edge if string $u$ gets transformed into $v$ by a rule from $R_1$. The analogous holds for $M_2$. Since the initial node of $M_1\odot M_2$ consists of a concatenation of $s_1$ and $s_2$, any node in the states graph of $M_1\odot M_2$ corresponds to some combination of a node of $M_1$ and one of $M_2$. Hence, the states graph of $M_1\odot M_2$ is the Cartesian product graph \cite{mathWorld4} of the states graphs of $M_1$ and $M_2$. Note how this is only possible because $M_1$ and $M_2$ are rule independent since otherwise more edges could be added due to rule matches overlapping between the $M_1$- and $M_2$-parts of the string or rules of one system getting applied to states of the other one.\par
To obtain the number of nodes reachable in this Cartesian product graph, we might without loss of generality first traverse a path of length $k$ in the \enquote{pure $M_1$-part} (i.\,e. the second half of the string is still $s_2$) and then take $n-k$ steps through the \enquote{pure $M_2$-part}. For the first subpath, we have $g_{M_1}(k)$ options by definition of the growth function. This gets multiplied by the $g_{M_2}(n-k)$ choices for the second subpath. Since we can choose $k$ freely, the resulting total count of nodes in the product graph is given by
\begin{equation}\label{eq:productSystem}
    g_{M_1\odot M_2}(n)=\sum_{k=0}^n g_{M_1}(k)\cdot g_{M_2}(n-k) = f(n)+g(n)+\sum_{k=1}^{n-1} g_{M_1}(k)\cdot g_{M_2}(n-k)
\end{equation}
where we set $g_{M_1}(0)=g_{M_2}(0)=1$ for convenience.\par 
For these two formulae, we have assumed the systems to be rule independent. For systems where this is not the case, they still give lower bounds of the combined systems growth rate since, generally speaking, in any system, the number of edges and nodes of the states graph can only remain constant or be increased when new rules are added. This might sound surprising as one could imagine \enquote{deletion rules} but rules that cause fewer rules being applied in the future do this only in newly added branches of the states graph (or not at all) not affecting the already existent graph. To make systems rule independent, we required the alphabets to be disjoint, however, this is not necessary as any multiway system can be emulated by a system over some binary alphabet so we can always make the alphabets of $M_1$ and $M_2$ equal.

\begin{lemma}\label{lemma:multiway2Symbol}
For any multiway system $M=(R,s,\Sigma)$, there is a multiway system $M'=(R',s',\{a, b\})$ where $a$ and $b$ are two distinct symbols, such that the states graphs of $M$ and $M'$ are isomorphic. 
\end{lemma}
\begin{proof}
To show this we will perform a \enquote{translation} from $M$ to $M'$, i.\,e. replace every symbol in $\Sigma$ by a word over $\{a,b\}$ using some bijection $f:\Sigma\rightarrow T \subset \{a,b\}^*$. By altering not only $s$ but also all rules, any word $w\in \Sigma^*$ matched by some rule in $R$ will correspond to the translated word in $w'\in T$ being matched by a rule in $R'$. Additionally, one must ensure that no two words in $T$ can overlap since otherwise, the rules could match in more places than before. Since there exist non-overlapping codes of arbitrary length, we can use these as elements in $T$ so there always exists some $f$ with the required properties. Thus, the actions of the rules on the states will be equal and isomorphic states graphs will be created. 
\end{proof}

Now, let us consider the algebraic properties of our multiway operations. We declare two multiway systems to be isomorphic (written $M_1 \cong M_2$), if and only if their states graphs are isomorphic. Isomorphic multiway systems always have equal growth functions. In the following analysis, we consider only the set of different equivalence classes of $\cong$, i.\,e. the set of all multiway systems up to isomorphism, and denote it by $\mathbb{M}$.\par 
Let $M_1,M_2,M_3\in \mathbb{M}$ be multiway systems and, without loss of generality, rule-independent. It is easy to see that $\oplus$ is commutative and associative since set unions are. More interestingly, the system $0_M:=(\{\},X,\{\})$ is a neutral element of $\oplus$ since 
\begin{align}
    M_1\oplus 0_M &= (R_1\cup \{\} \cup \{X\rightarrow s\mid s \in S_2(M_1)\cup \{\}\}, X, \Sigma_1\cup \{\}) \notag \\
    &= (R_1\cup \{X\rightarrow s\mid s\in S_2(M_1)\},X,\Sigma_1)
\end{align}
which is isomorphic to $(R_1,s_1,\Sigma_1)$ because the $X$ symbol is used only once, acting precisely as $s_1$ would have and therefore keeping the states graph structure unchanged.\par 
The commutativity of $\odot$ is granted because we required $M_1$ and $M_2$ to be rule-independent, so the order in which their states are concatenated does not matter because no overlaps where rules could apply on the intersection of $M_1$- and $M_2$-states can be created. Similarly, $\odot$ is associative, simply because string concatenation is. One can also prove both properties with the commutativity and associativity of the Cartesian graph product. There also is a neutral element of $\odot$, namely $1_M:=(\{\},``",\{\})$ or actually any system with no rules since its initial state will just be appended onto every state of the system one is multiplying with and the states graph will not change. \par 
Also notice that $\odot$ distributes from the left over $\oplus$:
\begin{align}
    M_1 &\odot (M_2\oplus M_3)\notag \\
    &= (R_1,s_1,\Sigma_1) \odot (R_2\cup R_3 \cup \{X\rightarrow S_2(M_2)\cup S_2(M_3)\}, X, \Sigma_2\cup \Sigma_3) \notag \\
    &= (R_1 \cup R_2\cup R_3 \cup \{X\rightarrow S_2(M_2)\cup S_2(M_3)\}, s_1 X, \Sigma_1\cup \Sigma_2\cup \Sigma_3) \notag \\ 
    &= ((R_1\cup R_2 \cup \{X\rightarrow S_2(M_2)\})\cup (R_1\cup R_3 \cup \{X\rightarrow S_2(M_3)\}), s_1 X, (\Sigma_1 \cup \Sigma_2) \cup (\Sigma_1 \cup \Sigma_3))\notag \\
    &= (R_1\cup R_2 \cup \{X\rightarrow S_2(M_2)\}, s_1, \Sigma_1 \cup \Sigma_2) \odot (R_1\cup R_3 \cup \{X\rightarrow S_2(M_3)\}, X, \Sigma_1 \cup \Sigma_3)\notag \\
    &= ((R_1,s_1,\Sigma_1)\oplus (R_2,s_2,\Sigma_2))\odot ((R_1,s_1,\Sigma_1)\oplus (R_3,s_3,\Sigma_3))\notag \\
    &= (M_1\oplus M_2)\odot (M_1\oplus M_3).
\end{align}
Thus, we get right distributivity from commutativity and conclude that $\odot$ distributes over $\oplus$. Therefore, we can conclude that $(\mathbb{M},\oplus,\odot)$ is a semiring, however with a weakened annihilation property which does not hold in general. As a consequence, their growth functions also form a semiring with weakened annihilation under the operations $(g_1 + g_2)(n):=g_1(n)+g_2(n)$ (with $(g_1+g_2)(1)$ defined to be $1$) and $(g_1 * g_2)(n) = \sum_{k=0}^n g_1(k)\cdot g_2(n-k)$ (with $g_1(0)=g_2(0):=1$). This demonstrates the potential of multiway systems to generate quite diverse and intricate growth functions as the two operations can be used to combine systems in various very interesting ways. The next section elaborates on this.

\subsection{Proof of theorem \ref{theorem:classesPartition}}\label{section:proofPartition}

Let us now construct multiway systems in the various growth classes to prove the above theorem \ref{theorem:classesPartition}. First, take the product system of a finite \enquote{chain}-system $M^N_1$ having one new state for $N$ generations until terminating and a constant system $M_2=(\{``A"\rightarrow ``AA"\}, ``A", \{``A"\})$. After the first $N-1$ steps, $M_1^N\odot M_2$ will have a constant growth function of value $N$ as the first $N$ terms in $\sum_{k=0}^n g_{M_1}(k)g_{M_2}(n-k)$ are one and all others zero because $M_1$ is finite. While this system produces (asymptotically) constant growth functions, it is not suited for multiplying arbitrary growth functions by constants. To achieve the latter, adding some system to itself several times resolves the issue.\par 
The next system to consider, $M^N_3$, is again given by the rule set $\{ ``A"\rightarrow ``AB"\}$ and started on a string of $N$ copies of $``A"$ denoted $``A^N"$. For calculating its growth function, we can represent it differently as the product system of $N$ instances of itself started on a single $``A"$ and thus having the growth rate $g^1_{M_3}(n)=1$. From this point of view, we can write the growth function of $M_3$ started on $``A^N"$ recursively as
\begin{equation}
    g^N_{M_3}(n)=\sum_{k=0}^n g^{N-1}_{M_3}(k)\cdot g^1_{M_3}(n-k) = \sum_{k=0}^n g^{N-1}_{M_3}(k)
\end{equation}
One might recognize this as the sequence of $(N-1)$-polytopic numbers \cite{mathWorld3}, also known as \enquote{figurate numbers}, of which the $n$-th element is given by $\binom{N-1+n-1}{N-1}$. Hence, $g^N_{M_3}$ is clearly a polynomial of degree $N-1$ and thus asymptotically equal to $x^{N-1}$.\par 
Now, consider $M^N_4=(\{``Q"\rightarrow ``Q x_i"\mid i=1,\dots,N\}$, $``Q"$, $\{Q, x_1, \dots, x_N\})$ for distinct symbols $x_i$. In the $n$-th step of evolution, it has basically generated all words of length $n$ over the alphabet of all $x_i$. Every node $``Q w"$ (where $w\in \{x_1,\dots,x_N\}^*$) in the states graph has $N$ outgoing edges to the nodes $``Q x_i w"$ for $1\leq i\leq N$. Thus, its growth function is precisely $g^N_{M_4}(n)=N^n$ allowing the possibility of multiway systems growing like all exponential functions. \par 
A more sophisticated example is the system $M^N_5=(\bigcup_{i=1,\dots,N} \{``T L"\rightarrow ``T x_i R", ``R T"\rightarrow ``L x_i T",``R x_i"\rightarrow ``x_i R", ``x_i L"\rightarrow ``L x_i"\},``T L T",\{``L",``R",``T",x_1,\dots,x_N\})$. Similarly to the previous system, $L$ and $R$ work as generators for words over the alphabet $\{x_1,\dots,x_N\}$ however only on the left and right ends of the word respectively. In every word ever produced by the system, there are exactly two $``T"$ symbols, one at the beginning and one at the end. After generating some new symbol between itself and the $``T"$, the generator $``L"$ or $``R"$ moves one step left or right respectively thereby not generating any new symbols as it is not next to a $``T"$. Since a new symbol is created every time a $``L"$ or $``R"$ reaches the respective $``T"$, the length of the word is increased every $n$ steps where $n$ the previous word length. Thus, their word-length is the sequence \enquote{$n$ occurs $n$ times} denoted $A_n$ and asymptotically equal to $(\sum_{k=0}^n k)^{-1} = (\frac{n(n+1)}{2})^{-1} = \frac{\sqrt{1+8 n}-1}{2} \in \Theta(\sqrt{n})$ by lemma \ref{lemma:inverseSummatory}. But the system's growth function is given by the number of possible words, i.\,e. $N^n$ and hence asymptotically equal to $N^{\sqrt{n}}$.\par 
The growth function of the previous system is noteworthy because it grows \enquote{intermediately} i.\,e. faster than every polynomial function and slower than all exponential functions. Formally, one checks this by noticing that $\lim_{x\rightarrow \infty} \frac{\ln(N^{\sqrt{x}})}{x} = 0$ (if the logarithm grows slower than $x$, the function is subexponential) and $\lim_{x\rightarrow \infty } \frac{\ln(N^{\sqrt{x}})}{\ln(x)} =\infty $ (if the logarithm grows faster than $\ln(x)$ times a constant, the function grows faster than $x^n$ for all $n$). In the study of groups and semigroups, which are related to multiway systems \cite{mwsNote2}, it has long been an open problem, finally solved by Grigorchuk \cite{gri1} to find groups of intermediate growth\footnote{Here, \enquote{growth} refers to the notion of \enquote{group growth rate} from group theory.}. For multiway systems, this turns out to be remarkably easy, supporting the claim that multiway growth functions are computationally diverse and powerful.\par 

\begin{figure}[ht]
\centering
\begin{subfigure}{\textwidth}
  \centering
  \includegraphics[width=.75\linewidth]{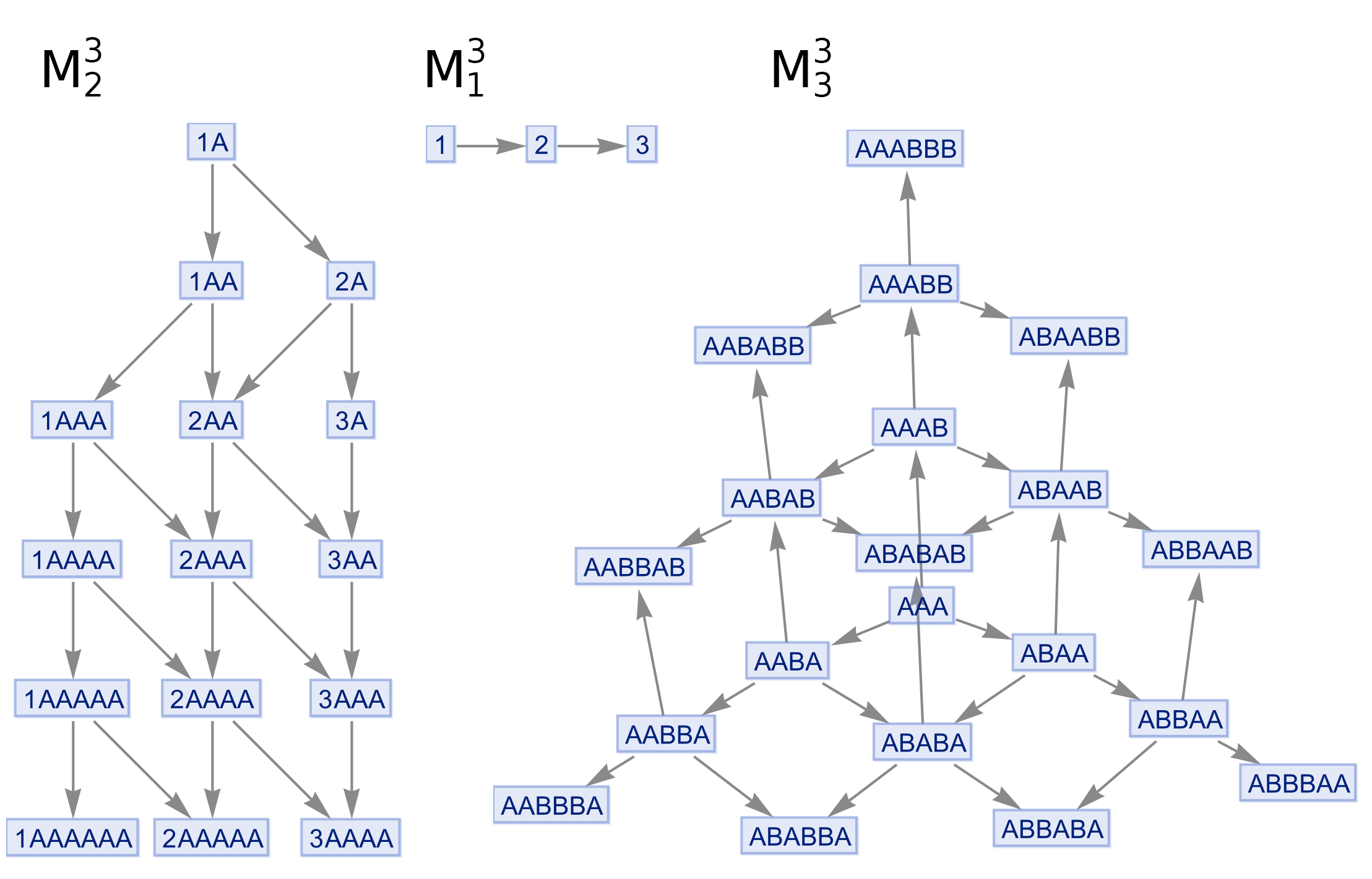}
\end{subfigure}\\
\begin{subfigure}{.45\textwidth}
  \centering
  \includegraphics[width=\linewidth]{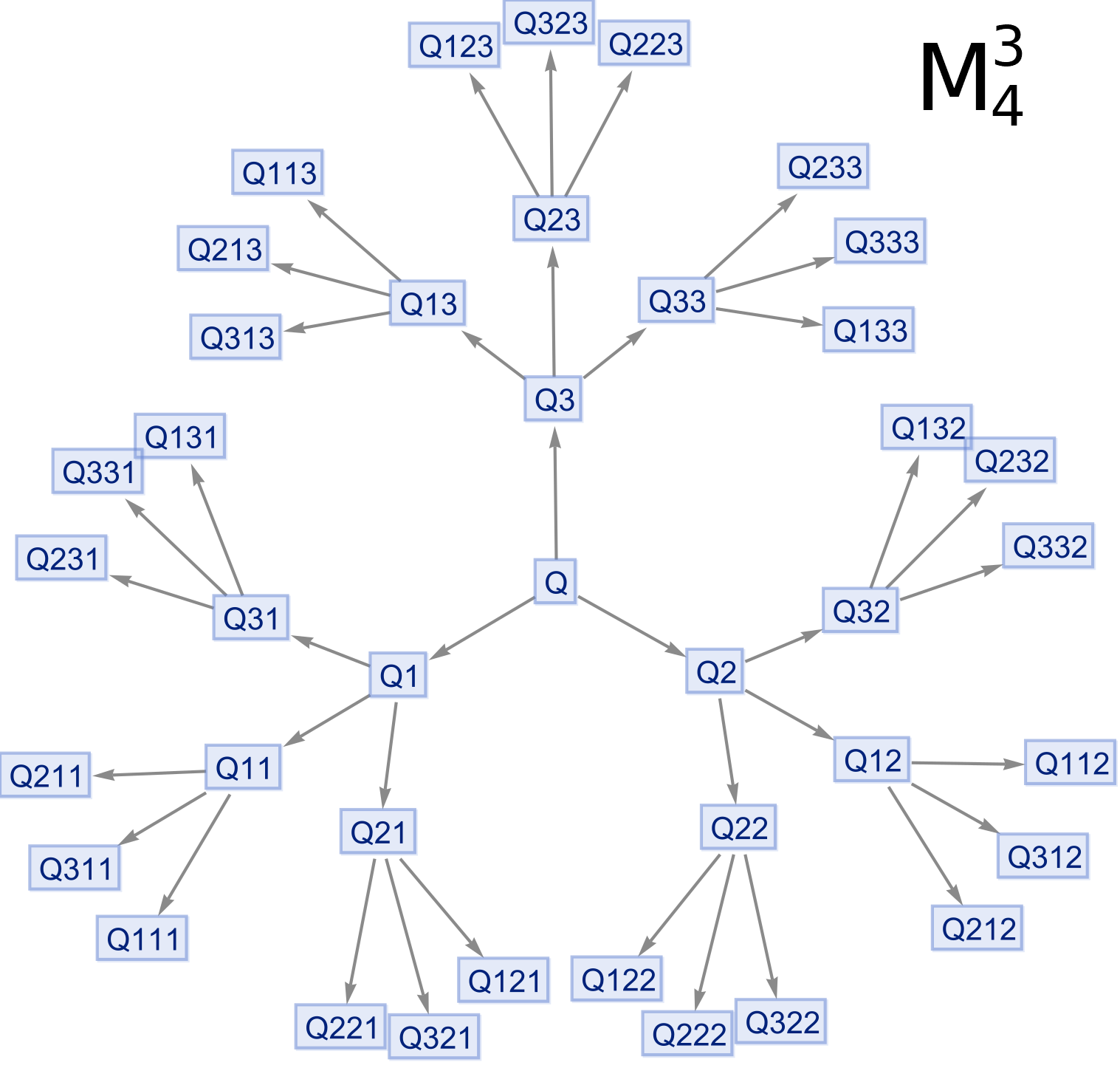}
\end{subfigure}
\hfill 
\begin{subfigure}{.45\textwidth}
  \centering
  \includegraphics[width=\linewidth]{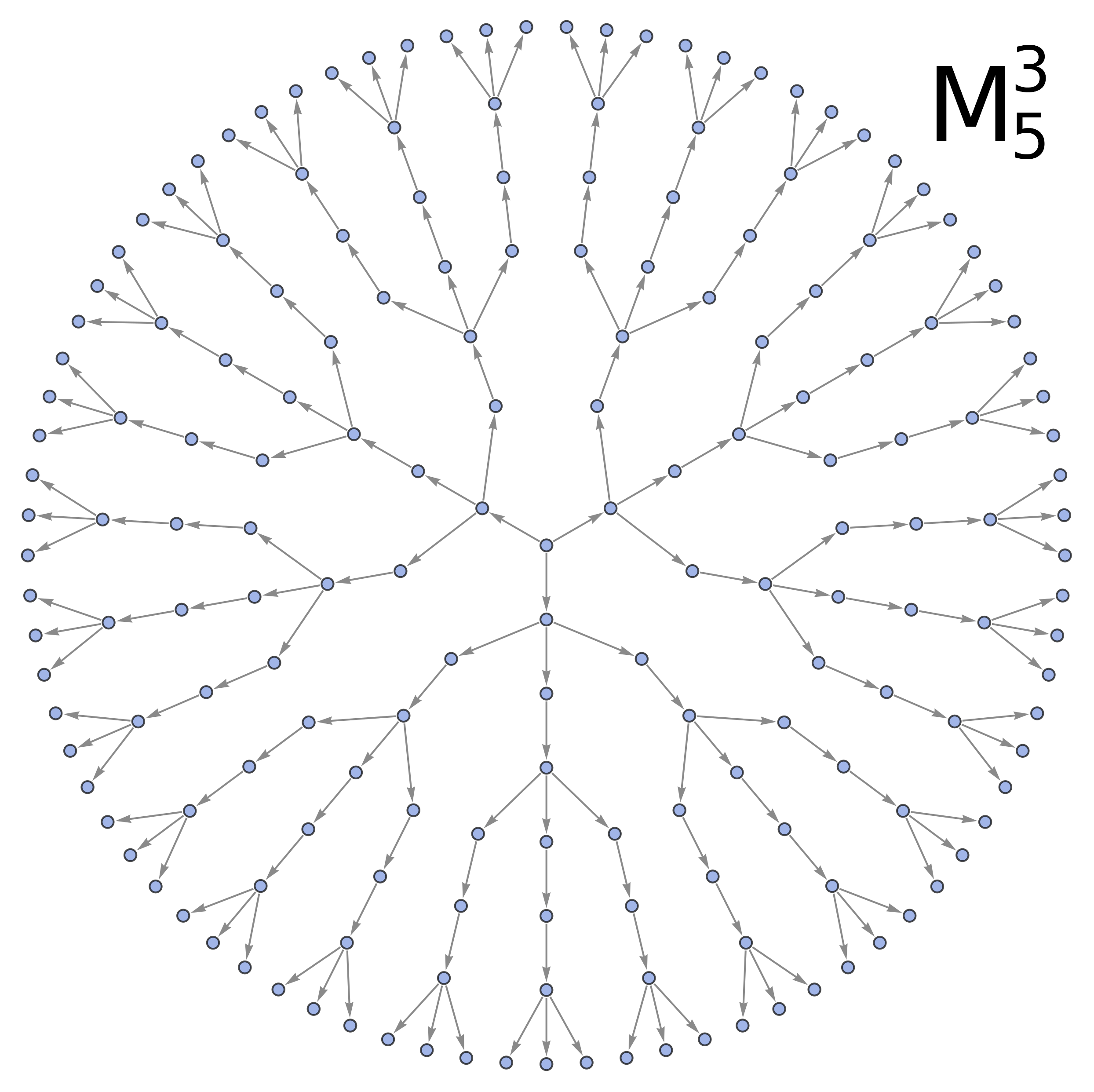}
\end{subfigure}
\caption{States graphs of exemplary systems $M_1^3$, $M_2^3$, $M_3^3$, $M_4^3$ and $M^3_5$ with growth functions asymptotically equal to $n\mapsto 0, n\mapsto 3, n\mapsto n^2, n\mapsto 3^n$ and $n\mapsto 3^{\sqrt{n}}$ respectively.}
\label{fig:m123Graphs}
\end{figure}

Now that we have shown the existence of infinitely many systems in $C_{\text{fin}}, C_{\text{bnd}}, C_{\text{pol}}, C_{\text{exp}}$ and $C_{\text{int}}$, only the classes of inverse functions remain. As mentioned above, theorem \ref{theorem:slowness} appears to be very useful for this. In section \ref{section:applicationSlowness}, we have already shown the existence of a system with a growth function asymptotically equal to $\log_2(x)$. The Turing machine from that example can be generalized to perform counting in any number system base $a$ yielding a halting function asymptotically equal to $n\mapsto a^n$ so that the same construction can be used to obtain multiway systems growing like $\log_a x$. It is not essential to go through the details here because all $\log_a x$ are asymptotically equal (since they only differ by constants by the base-change law). It is also clear that infinitely many systems in the class $C_{\text{invexp}}$ can be created because one could for example use multiway addition to add systems with constant growth functions to logarithmic systems. Notice how this shows generally that every class containing at least one system contains infinitely many systems.\par 
By the same style of argument, we conclude that there are infinitely many systems in $C_{\text{invpol}}$. Consider a weaker version $M_6$ of the system of intermediate growth rate $M^1_5$ for $N=1$. This system has one state forever but the length of its strings is still the sequence \enquote{$n$ occurs $n$ times}. Now simply apply the Turing machine construction from section \ref{section:slownessProof} to this system by adding the rules $``L T"\rightarrow ``Z"$, $``TR"\rightarrow ``Z"$ and $``Z"\rightarrow ``ZZ"$ respectively to generate different branches of the states graph. As described in section \ref{section:slownessProof}, this yields the desired growth function of $g_{M_6}(n)=A_n(n) \in \Theta(\sqrt{n})\Rightarrow M_6\in C_{\text{invpol}}$. \par 
Using corollary \ref{corollary:multiwaySlowerComputable}, it is also obvious that there are infinitely many multiway systems in $C_{\text{invsupexp}}$ as one might, for example, just construct a multiway system growing slower than the inverse Ackermann function. It remains to show that the class of inverses of intermediately growing functions $C_{\text{invint}}$ is non-empty. To show this, first note that there is a Turing machine $\mathcal{T}$ which computes $\floor{\sqrt{n}}$ when given $n$ in unary in polynomial time. If one feeds the result of this computation in the $\mathcal{T}_{\text{exp}}$ machine from section \ref{section:applicationSlowness}, one can construct a machine that halts after $T(n)=p(n)+2^{\floor{\sqrt{n}}}$ steps for some polynomial function $p(n)$. Similar to the argumentation from the proof of lemma \ref{lemma:analyticalInvert}, we see that $T_{\sigma}$ is asymptotically equal to the integral of $p(x)+2^{\sqrt{x}}$ given by $F(x)=q(x)+c\sqrt{x}2^{\sqrt{x}}$ where $q$ is some polynomial (because $\int 2^{\sqrt{x}}\d x=\frac{2^{\sqrt{x}+1}(\sqrt{x} \ln(2)-1)}{\ln^2(2)}$). Since $q$ is a polynomial, this function grows still intermediately which one can easily verify by calculating $\lim_{x\rightarrow \infty} \frac{\ln(F(x))}{\ln(x)}$ and $\lim_{x\rightarrow \infty} \frac{\ln(F(x))}{x}$. By theorem \ref{theorem:slowness} and lemma \ref{lemma:analyticalInvert}, there is a multiway system $M_7$ which has a growth function asymptotically equal to $F^{-1}(x)$. Since $F$ grows intermediately, this system is in $C_{\text{invint}}$. Together with the previous paragraphs, this provides the proof for theorem \ref{theorem:classesPartition}.

\subsection{Multiway Growth Approximability and Undecidability}\label{section:consequencesComputability}

As we have shown the existence of systems with growth function $g_M$ in $\Theta(1)$, $\Theta(x^n)$, $\Theta(a^x)$, $\Theta(a^{\sqrt{x}})$, $\Theta(\sqrt{x})$ and $\Theta(\ln(x))$, it is possible to find a multiway growth function asymptotically equal to any combination of these functions using point-wise addition and discrete convolution. This works because asymptotic equivalence is preserved under addition, summation and multiplication (\cite{hil1} section 2.2) so one can take the appropriate multiway systems for the atomic parts of the functions expression and combine them using the multiway sum and product. All of the basic functions are monotonously increasing and addition as well as discrete convolution preserve this property (because all functions have values in $\R_{\geq 0}$). This proves the following corollary to theorem \ref{theorem:classesPartition}:

\begin{corollary}\label{corollary:multiwayArithmetic}
If some function $f:\R_{\geq 0}\rightarrow \R_{\geq 0}$ there is a function $g\in \Theta(f)$ expressible as a finite combination of the functions $x\mapsto c, x\mapsto x^n, x\mapsto a^x, x\mapsto a^{\sqrt{x}}, x\mapsto \sqrt{x}, x\mapsto \ln(x)$ ($x, a, c\in \N_+$) using the operations point-wise addition and discrete convolution, then $f$ is multiway-growth-approximable.
\end{corollary}

The corollary provides further insights into the multiway-growth-approximability and multiway-growth-computability of functions. Let $\mathcal{M}_C$ and $\mathcal{M}_A$ be the sets of multiway-growth-computable and multiway-growth-approximable functions respectively. Note that $\mathcal{M}_C\subset \mathcal{M}_A$ and $\mathcal{M}_C$ is countably infinite, since the set of multiway systems is countably infinite\footnote{This follows from the fact that every multiway system can be reduced to use only the alphabet $\{A,B\}$ (lemma \ref{lemma:multiway2Symbol}) and then be expressed using the symbols $``A",``B",``\rightarrow",``\{",``\}",``(",``)"$ and $``,"$ by writing down its signature. However, the set of words over this finite 8-symbol alphabet is countably infinite.} whereas $\mathcal{M}_A$ is uncountably infinite because it contains, for example, all constant functions from $\R_{\geq 0}$ to $\R_{\geq 0}$. Using the above corollary, we notice that a variety of classes of functions are multiway-growth-approximable:
\begin{enumerate}
    \item For any function $f\in \mathcal{M}_C$, the function $\lambda \cdot f$ where $\lambda \in \N$ is a constant is also in $\mathcal{M}_C$.
    \item $\mathcal{M}_A$ contains all polynomials with natural coefficients because they can be build by adding powers $x^n$ multiplied by constants. 
    \item For all polynomials $p(x)$ in $\mathcal{M}_A$, $\mathcal{M}_A$ contains functions asymptotically equal to $x\mapsto \ln(x)\cdot p(x)$.
\end{enumerate}
Consequence 3 takes a little longer to prove but will be worth explaining in detail because similar methods may be used to generalise it, for example showing that $\mathcal{M}_A$ contains polylogarithmic functions. 

\begin{proof}
Consider the product system of a polynomial and a logarithmic system. By equation \ref{eq:productSystem} and the fact that asymptotic equivalence is preserved under addition and multiplication, the system's growth function is asymptotically equal to \begin{equation}
     n^a+\ln(n)+\sum_{k=1}^n k^a \ln(n-k) = \mathcal{O}(n^a) + \sum_{k=1}^{n-1} h(k).
\end{equation}
Consider some fixed input $n$ for $g$. Let $B_k$ be the $k$-th Bernoulli number, $B_m(x)$ the periodic continuation of the $m$-th Bernoulli polynomial and choose $m=n+1$. Using the Euler-Maclaurin Formula (\cite{storchWiebe}\,pp.\,501\,ff.), we obtain
\begin{equation}
    \sum_{k=1}^{n-1} h(k) = \int_{1}^{n-1} h(x)\d x  + \frac{h(n)+h(1)}{2} + S_m + R_m.
\end{equation}
First of all, $\frac{h(n)+h(1)}{2}$ simply evaluates to $\frac{\ln(n)}{2}\in \mathcal{O(1)}$. Next, it is easy to show inductively that the $k$-th derivative of $h$ is of the form $h^{(k)}(x)=x^{a-k}(c\ln(n-x)+p(x))$ for a real constant $c$ and a rational function $p(x)\in \mathcal{O}(1)$ as long as $k\geq a$. The $a+1$-th derivative is some rational function in $\mathcal{O}(\frac{1}{x})$. Thus, the remainder sum and integral satisfy
\begin{align}
    S_m &=\sum_{k=2}^{n+1} \frac{(-1)^k B_{k}}{k!}(h^{(k)}(n) - h^{(k)}(1)) \sim_\Theta \sum_{k=2}^{n+1} x^{a-k} \ln(n-x) \in \mathcal{O}(x^{a-2} \ln(n-x)) \\
    R_m &= \frac{(-1)^{n+2}}{(n+1)!}\int_1^{n-1} f^{(n+1)}(x) B_{m+1}(x) \d x \sim_\Theta \int_1^{n-1} \mathcal{O}(\frac{1}{x})\d x \in \mathcal{O}(\ln(x)).
\end{align}
Hence, $\sum_{k=1}^{n-1} h(k) \sim_\Theta \int_1^{n-1} h(x)\d x$. Using the fact that $\diffrac{}{x}(x+(n-x)\ln(n-x)) = -\ln(n-x)$ and applying integration by parts, we obtain
\begin{align}
    I_a &=
    \int x^a \ln(n-x)\d x = -x^a(x+(n-x)\ln(n-x)) 
    + \int (x+(n-x)\ln(n-x))ax^{a-1}\d x = T + J \notag \\
    J &= a\int x^a \d x + na \int x^{a-1}\ln(n-x) \d x - a\int x^{a} \ln(n-x) \d x = \frac{a}{a+1}x^{a+1} + na I_{a-1} - a I_a \notag \\ 
    &\Rightarrow (a+1)I_a = -x^a(x+(n-x)\ln(n-x)) + \frac{a}{a+1}x^{a+1} + na I_{a-1} + C \notag \\
    &\Rightarrow \int_1^{n-1} x^a \ln(n-x)\d x = I_a(n-1)-I_a(1) \\
    &=\frac{1}{a+1}(-(n-1)^a(n-1+1\cdot 0) + \frac{a}{a+1}(n-1)^{a+1} + n a I_{a-1}(n-1) +\notag \\
    &\quad 1^a(1+(n-1)\ln(n-1)) - \frac{a}{a+1}1^{a+1} - na I_{a-1}(1) ) \notag\\
    &=\frac{1}{a+1}\left( -\frac{(n-1)^{a+1}}{a+1} +\frac{1}{a+1} +(n-1)\ln(n-1) + na( I_{a-1}(n-1) -I_{a-1}(1))  \right)\notag\\
    &=\frac{1}{a+1} \left( \frac{1-(n-1)^{a+1}}{a+1} +(n-1)\ln(n-1) + na\int_1^{n-1} x^{a-1}\ln(n-x)\d x \right).
\end{align}
For $a=1$ we know 
\begin{align}
    \int_1^{n-1} x^{a-1}\ln(n-x)\d x &= [-(x+(n-x)\ln(n-x))]_1^{n-1} = -(n-1)+(1+(n-1)\ln(n-1))\notag\\
    &= -n+2+n\ln(n-1) -\ln(n-1) \in \Theta(n\ln(n-1))\notag \\ 
    \Rightarrow \int_1^{n-1} x^a\ln(n-x)\d x &= \frac{1}{2}(\frac{1-(n-1)^2}{2} +(n-1)\ln(n-1) + n\Theta(n\ln(n-1)) \notag \\
    &\Rightarrow \int_1^{n-1} x^1\ln(n-x)\d x \in \Theta(n^2\ln(n-1)).
\end{align}
Now, by assuming $\int_1^{n-1} x^a\ln(n-x)\d x \in \Theta(n^{a+1}\ln(n-1))$
we have \begin{align}
    \int_1^{n-1} x^a\ln(n-x)\d x &= \frac{1}{a+1}\left( \frac{1-(n-1)^{a+1}}{a+1} +(n-1)\ln(n-1) + na\int_1^{n-1} x^{a-1}\ln(n-x)\d x  \right) \notag \\
    & = \Theta(-(n-1)^{a+1}) + \Theta((n-1)\ln(n-1)) + an\Theta(n^{a}\ln(n-1)) \notag\\
    \Rightarrow \int_1^{n-1} x^a\ln(n-x)\d x &\in \Theta(n^{a+1}\ln(n-1))
\end{align}
proving $\int_1^{n-1} x^a\ln(n-x)\d x \in \Theta(n^{a+1}\ln(n-1))$ inductively. Finally, this means the growth function of our multiway system is asymptotically equal to \begin{equation}
     \mathcal{O}(n^a) + \sum_{k=1}^{n-1} h(k) \sim_\Theta n^{a+1}\ln(n-1) \sim_\Theta n^{a+1}\ln(n)
\end{equation}
proving in fact, that for every $x^a$-system, there exists a $x^a\ln(x)$-system.
\end{proof}

From these three properties, we might already conclude that a significant number of functions usually investigated in mathematical analysis can be approximated by multiway growth functions. This demonstrates the computational diversity of multiway growth functions which is neither a trivial nor an expected property. \par 
We know that multiway systems themselves are capable of universal computation as they can emulate Turing machines but it is unknown so far which computations their growth functions are able to perform. Thus, it may be considered remarkable that many common mathematical functions are expressible (i.\,e. approximable) as multiway growth functions. Conversely, this could later allow to make statements about a multiway system's complexity or structure by considering only its growth function. Maybe, multiway systems can even be used to make general statements about the mathematical functions themselves since they give a new way of looking at them. \par 
Notice however that multiway growth functions are strictly less powerful than computable functions in general due to the following lemma:
 
\begin{lemma}
Every multiway growth function is primitively recursive.
\end{lemma}
\begin{proof}
Given some multiway system, one can compute the growth function $g(n)$ in the following way: One uses two lists $S_n$ and $T_n$ ($S_0=\{\}, T_0=\{s\}$) to store all states the system has had until generation $n$ and all states of generation $n$. In every iteration, the length of $T_n$ is the value of $g(n)$. In the $n+1$-th step, the algorithm iterates through all rules and searches through the characters of all strings in $T_n$ to check if any rule applies. If this happens, the string with some part replaced will be added to $T_{n+1}$ if it is not in $S_n$ or $T_{n+1}$ already (only new states get added). After all possible such operations are done, $T_{n+1}$ contains all new states of the system and we set $S_{n+1}$ to be $S_n \cup T_{n+1}$. When repeated, this process simulates the system's evolution and thus yields the correct $g(n)$. All loops required can be implemented using DO-loops. If strings are treated as lists of characters (numbers), the string replacement and substring matching operations can be implemented using only list insertions, deletions and searches through the list. No further data structures and no comparisons are needed. Hence, the entire program is primitively recursive by \cite{mey1}. 
\end{proof}

As primitive recursive functions still contain some superexponential functions, multiway growth functions are also strictly weaker than those. Still, multiway growth functions can approximate a lot of elementary functions so they might even be stronger than elementary arithmetic (EA) while probably weaker than EA+ (EA and the axiom that the superexponential function is total). It remains an open question to find out how $\mathcal{M}_A$ or $\mathcal{M}_C$ can be characterized elegantly.\par
Besides their use for investigating multiway-growth-computability, the tools obtained in this paper allow us to elegantly prove some statements about undecidability of multiway-growth-related questions. For example, deciding if a given multiway system is finite is undecidable as a system could simulate some arbitrary Turing machine and have zero new states when the machine halts, reducing the question to the halting problem. Additionally, even for an infinite multiway system, deciding whether its growth function is equal to some conjectured function is undecidable in general since the system could be the sum of some usual system and a Turing machine emulator which becomes, for example, exponentially growing after the machine halts but has only one state before that. This observation makes it especially important in the context of the Wolfram Physics Project to not only use empirical (computed) observations about a system's growth function, rate or class but also take into account the system's rule when conjecturing about it.\par
This \enquote{trick} of integrating a system which grows very differently once some Turing machine halts into some larger system was also the strategy used to construct the strongly oscillating system in figure \ref{fig:notSameThetaBounds}. Specifically, we first construct a system similar to $M^N_5$ from section \ref{section:proofPartition}, defined as $M_8^{N,M}=(\bigcup_{i=1,\dots,N} \{``R A"\rightarrow ``x_i R"\},``A^M",\{``R", A,x_1,\dots,x_N\}$. When started on a string of $M$ $A$s, $R$ moves to the right while replacing the $A$ it just moved over by any of the $N$ $x_i$. Hence, the states of the system after $N$ steps are precisely the words of length $N$ over $\{x_1,\dots,x_N\}$ since $R$ moved $N$ steps to the right. When $R$ reaches the right end, the system terminates. Thus, this system has $M^N$ new states after $M$ steps and $0$ after that.\par  
The oscillating system from figure \ref{fig:notSameThetaBounds} is now the sum of a linearly growing system and a version the logarithm system from section \ref{section:applicationSlowness}. However, the customized logarithm system does not increase its number of states after $\mathcal{T}_{\text{exp}}$ halts but just triggers an instance of $M^{4,N}_8$ on the string of ones $\mathcal{T}_{\text{exp}}$ has written. Since after roughly $2^n$ steps, the $n$-th version of $\mathcal{T}_{\text{exp}}$ halts, $n$ ones are written on the tape so about $n$ steps later, the system has $4^n$ states for one generation and then \enquote{collapses} into one state again. The smallest monotonically increasing upper bound for this growth function is one that stays $4^n=(2^n)^2$ for roughly $2^n$ steps and then increases to $4^{n+1}$. Denote this sequence $\overline{a}(n)$. Now $\log_2{\sqrt{\overline{a}(n)}}$ is approximately the sequence \enquote{$n$ occurs $2^n$ times} which is in $\Theta(\log_2(n))$. Thus, $\log_2{\sqrt{\overline{a}(n)}}\in \Theta(\log_2(n))\iff \sqrt{\overline{a}(n)}\in \Theta(n)\iff \overline{a}(n)\in \Theta(n^2)$. Since the whole system consists of this and an added linear system, the total growth rate has lower and upper tight bounds of $\Theta(x)$ and $\Theta(x^2)$ respectively. If one generalizes the methods used in this paper, they might be used perform a kind of \enquote{multiway system engineering} i.\,e. they could help to construct systems for specific purposes.

\section{Concluding Remarks}\label{section:conclusion}

Our main results may be summarised as follows:
\begin{enumerate}
    \item This paper introduced the formalisms of multiway growth functions, rates and classes which have large potential for mathematically investigating multiway systems.
    \item In theorem \ref{theorem:slowness}, we showed that multiway systems can grow slower than all computable functions while never exceeding exponential functions. Not only is this asymmetry very suspicious of a more general underlying principle but the theorem also demonstrates that multiway growth functions cannot be trivial and must have some computational complexity associated with them.
    \item This gets supported by the fact that multiway systems are capable of simulating (approximating) quite an extend of known functions while being contained in the set of primitive recursive functions. It remains entirely unclear, which status the sets $\mathcal{M}_A$ and $\mathcal{M}_C$ have among other well-known sets of functions but our theorem \ref{theorem:classesPartition} starts a characterisation by subdividing them into nontrivial classes. 
    \item Additionally, we have exemplarily demonstrated some basic systems which can be combined to yield systems giving a wanted growth function. This \enquote{multiway engineering} could be generalised and turn out to be useful for getting intuition about multiway systems as well as potentially constructing (counter-)examples to empirically grounded conjectures existing in the Wolfram Physics Project.
    \item Another very interesting foundation for further research are the arithmetic-like operations on multiway systems which we have shown to equip the set of all multiway systems with an almost-semiring structure. 
\end{enumerate}
These results could be applied in various ways:
\begin{enumerate}
    \item The most obvious next step is generalising our theorems to hypergraphs to make them meaningful for the actual Wolfram Model. However, this should not be difficult to do. Generalising the algebraic structure of the set of multiway systems with the operations we defined seems much more interesting, especially for hypergraph-based multiway systems since it becomes relevant in the theoretical physics context of multiway systems as there are various recent findings about a connection of Wolfram Models, category theory and quantum mechanics \cite{zxPaper}. Additionally, more general forms of our results could be obtained in context of the connection between multiway systems and the foundations of homotopy type theory \cite{xerxesSummerSchool}.
    \item More specifically, the changes in structure of branchial space (see glossary of \cite{zxPaper}) over time and thus the growth functions of multiway systems are related to the states of quantum systems and measurements of these \cite{gor1}. Our upper bound on multiway system growth rate may be used to give upper bounds on entanglement speed and maximum possible information entropy in Wolfram Models. Similarly, the fact that slowness of multiway growth rates is \enquote{unbounded} could be used to investigate very slowly developing quantum systems which might be especially stable and hence useful for quantum computation, but that is mere speculation. More straightforward is the application of multiway growth classes to estimating the complexity of quantum computational algorithms or make predictions about quantum supremacy using the Wolfram model (c.\,f. \cite{summerSchoolShor} and \cite{zxPaper}).
    \item Another potential physical application is early-universe cosmology. Since there seems to be an empirical connection \cite{summerSchoolManojna}, formalising which would also be an important project, between the growth rates of physical and branchial space i.\,e. in our context string-length (or hypergraph size) and multiway growth functions, our boundaries on growth functions and especially the \enquote{unboundedness} of its slowness may be related to the physical expansion of the early universe in the view of the Wolfram Model's formalism. 
    \item Building on our classification scheme, for example by allowing combined classes like \enquote{polynomial times inverse intermediate}-functions, trying to quantify and find regularity in the oscillations of a system's growth function (consider the system in figure \ref{fig:tightBounds} as an example) or analysing strongly oscillating systems like the one depicted in figure \ref{fig:notSameThetaBounds}. 
    \item Related to the regularity of multiway systems, one can ask whether determining which growth rate (or growth class) a given multiway system has from its rules is possible in general. Maybe, methods from automated theorem proving could be used for this. Another question is whether there are multiway systems which show no regularity in their growth functions. The latter would be important for \enquote{A New Kind of Science}-related research since a considerable part of Wolfram's work concerns the complexity and irregularity of such computational systems \cite{wol2}.
\end{enumerate}
The preceding points are just a few possibilities showing how much potential the investigation of multiway systems and their growth functions, rates and classes has. This paper marks only the beginning of many further research projects. However, while we lay a very basic foundation, we succeed in doing so as our results are formally proven and computationally applicable.

\section{Acknowledgements}

First and foremost, I would like to thank Stephen Wolfram for suggesting this project and giving important advice concerning the general directions of research and methodology. Speaking of methodology, I have to mention my mentor Xerxes D. Arsiwalla for whose guidance and feedback I am very grateful. Many thanks also to Jonathan Gorard for repeatedly proof-reading this paper and assisting the proof that multiway systems form a semiring with weakened annihilation property (section \ref{sec:multiwayArithmetic}). Additionally, I highly appreciated the encouragement and support of Peter Barendse who helped to kickstart this research project at the Wolfram Summer Camp 2020 and Paul Siewert whose critique and explanations were very useful for formalising the proofs presented in this paper. 
\newpage 
\printbibliography

\typeout{get arXiv to do 4 passes: Label(s) may have changed. Rerun}
\end{document}